\documentclass[twoside,leqno,twocolumn]{article}
\usepackage{ltexpprt}
\usepackage[algo2e, ruled, vlined, nofillcomment, linesnumbered]{algorithm2e}
\usepackage{amsopn}
\usepackage{mathtools}
\usepackage{amsmath}
\usepackage{float}
\usepackage{balance}

\usepackage{xspace}
\usepackage{amsfonts}
\usepackage{cite}

\newtheorem{problem}{Problem}
\newtheorem{observation}{Observation}

\newcommand{\spara}[1]{\smallskip\noindent{\bf #1}}

\newcommand{\graph}{\ensuremath{G}\xspace}
\newcommand{\nodes}{\ensuremath{V}\xspace}
\newcommand{\edges}{\ensuremath{E}\xspace}
\newcommand{\nnodes}{\ensuremath{n}\xspace}
\newcommand{\nedges}{\ensuremath{m}\xspace}
\newcommand{\seed}{\ensuremath{s}\xspace}
\newcommand{\rep}{\ensuremath{R}\xspace}
\newcommand{\nrep}{\ensuremath{k}\xspace}
\newcommand{\path}{\ensuremath{p}\xspace}
\newcommand{\pathtwo}{\ensuremath{q}\xspace}

\newcommand{\cost}{\ensuremath{c}\xspace}
\newcommand{\tree}{\ensuremath{T}\xspace}
\newcommand{\exclpath}{\ensuremath{Q}\xspace}
\newcommand{\excldist}{\ensuremath{d}\xspace}
\newcommand{\extexclpath}{\ensuremath{Q^{+}}\xspace}
\newcommand{\extexcldist}{\ensuremath{d^{+}}\xspace}

\newcommand{\ts}{\ensuremath{\mathcal{T}}\xspace}
\newcommand{\tmin}{\ensuremath{t_0}\xspace}
\newcommand{\bigO}{\ensuremath{\mathcal{O}}\xspace}
\newcommand{\terminals}{\ensuremath{X}\xspace}

\newcommand{\nphard}{{{\np}-hard}\xspace}
\newcommand{\set}[1]{\left\{#1\right\}}
\newcommand{\floor}[1]{\left\lfloor{#1}\right\rfloor}
\newcommand{\pr}[1]{\left(#1\right)}

\newcommand{\abs}[1]{{\left|#1\right|}}

\newcommand{\real}{\mathbb{R}}
\newcommand{\np}{\textbf{NP}\xspace}

\newcommand{\define}{\leftarrow}

\newcommand{\Ourproblem}{{\sc Order\-ed\-Steiner\-Tree}\xspace}
\newcommand{\routedproblem}{{\sc Order\-ed\-Steiner\-Tree}\ensuremath{(\seed)}\xspace}
\newcommand{\STproblem}{{\sc Steiner\-Tree}\xspace}

\newcommand{\bfs}{\textsc{bfs}\xspace}
\newcommand{\algclosure}{\textsc{closure}\xspace}
\newcommand{\algmst}{\textsc{delayed-bfs}\xspace}
\newcommand{\alggreedy}{\textsc{greedy}\xspace}

\newcommand{\algone}{\algclosure}

\newcommand{\tbfs}{\algmst}
\newcommand{\baseline}{{\sc steiner}\xspace}

\newcommand{\si}{\textsc{si}\xspace}
\newcommand{\seir}{\textsc{seir}\xspace}
\newcommand{\mdl}{\textsc{mdl}\xspace}
\newcommand{\ic}{\textsc{ic}\xspace}
\newcommand{\ct}{\textsc{ct}\xspace}
\newcommand{\spath}{\textsc{sp}\xspace}

\newcommand{\grqc}{\textit{grqc}\xspace}
\newcommand{\email}{\textit{email-eu}\xspace}
\newcommand{\arx}{\textit{arxiv-hep-th}\xspace}
\newcommand{\fb}{\textit{facebook}\xspace}

\newcommand{\enron}{\textit{enron}\xspace}
\newcommand{\digg}{\textit{Digg}\xspace}

\makeatletter
\newcommand{\nosemic}{\renewcommand{\@endalgocfline}{\relax}}
\newcommand{\dosemic}{\renewcommand{\@endalgocfline}{\algocf@endline}}
\newcommand{\pushline}{\Indp}
\newcommand{\popline}{\Indm\dosemic}
\@ifpackageloaded{algorithm2e}{
	\SetKwComment{tcpas}{\{}{\}}
	\SetCommentSty{textnormal}
	\SetKw{Null}{null}
	\SetKw{AND}{and}
	\SetKw{OR}{or}
	\SetKw{EACH}{each}
	\SetKwInOut{Input}{input}
	\SetKwInOut{Output}{output}
	\SetKwInput{KwInput}{Input}
	\SetKwInput{KwOutput}{Output}
	\SetArgSty{textnormal}
	\SetKw{False}{false}
	\SetKw{True}{true}
	\SetKw{Break}{break}
	\SetKw{None}{none}
	\SetKw{Continue}{continue}
}{}
\makeatother

\newenvironment {squishlist}
{\begin{list}{$\bullet$}
		{ \setlength{\itemsep}{0pt}
			\setlength{\parsep}{3pt}
			\setlength{\topsep}{3pt}
			\setlength{\partopsep}{0pt}
			\setlength{\leftmargin}{1.5em}
			\setlength{\labelwidth}{1em}
			\setlength{\labelsep}{0.5em} } }
	{\end{list}}

\usepackage{tikz,pgfplots,pgfplotstable}
\usetikzlibrary{positioning,fit,shapes}
\usetikzlibrary{decorations.pathreplacing}
\usetikzlibrary{arrows}

\pgfdeclarelayer{background}
\pgfdeclarelayer{foreground}
\pgfsetlayers{background,main,foreground}

\makeatletter
\tikzset{multicircle/.style  args={#1, #2}{%
 alias=tmp@name, %
  postaction={%
    insert path={
     \pgfextra{%
     \pgfpointdiff{\pgfpointanchor{\pgf@node@name}{center}}%
                  {\pgfpointanchor{\pgf@node@name}{east}}%
     \pgfmathsetmacro\insiderad{\pgf@x}%
        \fill[white] (\pgf@node@name.center)  circle (\insiderad-\pgflinewidth);%
        \draw[#2] (\pgf@node@name.center)  circle (\insiderad-\pgflinewidth);%
        \fill[#2] (\pgf@node@name.center)  -- ++(0:\insiderad-\pgflinewidth) arc (0:#1:\insiderad-\pgflinewidth)--cycle;%
        }}}}}
\makeatother

\definecolor{yafaxiscolor}{rgb}{0.3, 0.3, 0.3}
\definecolor{yafcolor1}{rgb}{0.4, 0.165, 0.553}
\definecolor{yafcolor2}{rgb}{0.949, 0.482, 0.216}
\definecolor{yafcolor3}{rgb}{0.47, 0.549, 0.306}
\definecolor{yafcolor4}{rgb}{0.925, 0.165, 0.224}
\definecolor{yafcolor5}{rgb}{0.141, 0.345, 0.643}
\definecolor{yafcolor6}{rgb}{0.965, 0.933, 0.267}
\definecolor{yafcolor7}{rgb}{0.627, 0.118, 0.165}
\definecolor{yafcolor8}{rgb}{0.878, 0.475, 0.686}
\definecolor{yafcolor9}{rgb}{0.965, 0.733, 0.767}

\newlength{\yafaxispad}
\newlength{\yaftlpad}
\newlength{\yaflabelpad}
\newlength{\yafaxiswidth}
\newlength{\yafticklen}
\makeatletter
\def\pgfplots@drawtickgridlines@INSTALLCLIP@onorientedsurf#1{}
\makeatother

\newcommand{\yafdrawxaxis}[2]{
  \pgfplotstransformcoordinatex{#1}\let\xmincoord=\pgfmathresult 
  \pgfplotstransformcoordinatex{#2}\let\xmaxcoord=\pgfmathresult 
  \pgfsetlinewidth{\yafaxiswidth} 
  \pgfsetcolor{yafaxiscolor}
  \pgfpathmoveto{\pgfpointadd{\pgfpointadd{\pgfplotspointrelaxisxy{0}{0}}{\pgfqpointxy{\xmincoord}{0}}}{\pgfqpoint{-0.5\yafaxiswidth}{\yafaxispad}}}
  \pgfpathlineto{\pgfpointadd{\pgfpointadd{\pgfplotspointrelaxisxy{0}{0}}{\pgfqpointxy{\xmaxcoord}{0}}}{\pgfqpoint{0.5\yafaxiswidth}{\yafaxispad}}}
  \pgfusepath{stroke}

}
\newcommand{\yafdrawyaxis}[2]{
  \pgfplotstransformcoordinatey{#1}\let\ymincoord=\pgfmathresult 
  \pgfplotstransformcoordinatey{#2}\let\ymaxcoord=\pgfmathresult 
  \pgfsetlinewidth{\yafaxiswidth} 
  \pgfsetcolor{yafaxiscolor}
  \pgfpathmoveto{\pgfpointadd{\pgfpointadd{\pgfplotspointrelaxisxy{0}{0}}{\pgfqpointxy{0}{\ymincoord}}}{\pgfqpoint{\yafaxispad}{-0.5\yafaxiswidth}}}
  \pgfpathlineto{\pgfpointadd{\pgfpointadd{\pgfplotspointrelaxisxy{0}{0}}{\pgfqpointxy{0}{\ymaxcoord}}}{\pgfqpoint{\yafaxispad}{0.5\yafaxiswidth}}}
  \pgfusepath{stroke}
}

\pgfplotscreateplotcyclelist{yaf}{%
{yafcolor1,mark options={scale=0.75},mark=o}, 
{yafcolor2,mark options={scale=0.75},mark=square},
{yafcolor3,mark options={scale=0.75},mark=triangle},
{yafcolor4,mark options={scale=0.75},mark=o},
{yafcolor5,mark options={scale=0.75},mark=o},
{yafcolor6,mark options={scale=0.75},mark=o},
{yafcolor7,mark options={scale=0.75},mark=o},
{yafcolor8,mark options={scale=0.75},mark=o}} 

\pgfplotsset{axis y line=left, axis x line=bottom,
  tick align=outside,
  compat = 1.3,
  tickwidth=\yafticklen,
  clip = false,
  every axis title shift = 0pt,
    x axis line style= {-, line width = 0pt, opacity = 0},
    y axis line style= {-, line width = 0pt, opacity = 0},
    x tick style= {line width = \yafaxiswidth, color=yafaxiscolor, yshift = \yafaxispad},
    y tick style= {line width = \yafaxiswidth, color=yafaxiscolor, xshift = \yafaxispad},
    x tick label style = {font=\scriptsize, yshift = \yaftlpad},
    y tick label style = {font=\scriptsize, xshift = \yaftlpad},
    every axis y label/.style = {at = {(ticklabel cs:0.5)}, rotate=90, anchor=center, font=\scriptsize, yshift = -\yaflabelpad},
    every axis x label/.style = {at = {(ticklabel cs:0.5)}, anchor=center, font=\scriptsize, yshift = \yaflabelpad},
    x tick label style = {font=\scriptsize, yshift = 1pt},
    grid = major,
    major grid style  = {dash pattern = on 1pt off 3 pt},
  every axis plot post/.append style= {line width=\yafaxiswidth} ,
  legend cell align = left,
  legend style = {inner sep = 1pt, cells = {font=\scriptsize}},
  legend image code/.code={%
    \draw[mark repeat=2,mark phase=2,#1] 
    plot coordinates { (0cm,0cm) (0.15cm,0cm) (0.3cm,0cm) };%
  } 
}

\begin{document}


\fancyfoot[R]{\footnotesize{\textbf{Copyright \textcopyright\ 2018 by SIAM\\
      Copyright for this paper is retained by authors}}}

\title{Reconstructing a cascade from temporal observations}
\author{
Han Xiao$^{\star}$ \quad Polina Rozenshtein$^{\star}$ \quad Nikolaj Tatti$^\dagger$ \quad Aristides Gionis$^{\star}$ \\
\begin{tabular}{cc}
$^{\star}$Aalto University & $^\dagger$F-Secure \\
Espoo, Finland & Helsinki, Finland \\
\texttt{firstname.lastname@aalto.fi} & \texttt{nikolaj.tatti@f-secure.fi}
\end{tabular}
}
\date{}

\maketitle





\begin{abstract}
Given a subset of active nodes in a network
can we reconstruct the cascade that has generated these observations?
This is a problem that has been studied in the literature, 
but here we focus in the case that 
temporal information is available about the active nodes. 
In particular, 
we assume that in addition to the subset of active nodes
we also know their activation time. 

We formulate this cascade-reconstruction problem 
as a variant of a Steiner-tree problem:
we ask to find a tree that spans all reported active nodes
while satisfying temporal-consistency constraints.
%
We present three approximation algorithms. 
The best algorithm in terms of quality
achieves a $\bigO(\sqrt{\nrep})$-approximation guarantee, 
where \nrep is the number of active nodes, 
while the most efficient algorithm has linear\-ithmic running time, 
making it scalable to very large graphs.

We evaluate our algorithms on real-world networks with both simulated and real cascades. 
Our results indicate that utilizing the available temporal information 
allows for more accurate cascade reconstruction.
Furthermore, our objective leads to finding the ``backbone'' of the cascade
and it gives solutions of high precision.

\end{abstract}

\section{Introduction}

Ideas, behaviors, computer viruses, and diseases, spread in networks. 
People are 
influencing each other
adopting behaviors, innovations, or memes.
Diseases like flu or measles are transmitted from person to person, 
while computer viruses spread in computer networks.
The study of diffusion processes has been a central theme in network science and graph mining.
Topics of interest include modeling diffusion processes~\cite{gomez-rodriguez2013modeling,kempe2003maximizing,kermack1927contribution}, 
devising strategies to contain the spread of epidemics~\cite{pastor2002immunization,prakash2010virus},
maximizing the spread of influence for marketing purposes~\cite{kempe2003maximizing},
as well as identifying starting points and missing nodes in cascades~\cite{feizi2014network,farajtabar2015back,lappas2010finding,prakash:12:netsleuth,rozenshtein2016reconstructing,sadikov:2011,shah:11:culprit,sundareisan2015hidden}.

In this paper 
we consider the problem of reconstructing a cascade
that has occurred in a network, given partial observations.
We model the problem by considering a graph 
$\graph=(\nodes, \edges)$
and a subset of nodes $\nodes'\subseteq\nodes$
that have been activated/infected during the cascade
(people who have fallen sick, users who have adopted an innovation, etc.).
The goal is to 
infer the hidden cascade
and reconstruct
other possible active nodes.
Previous approaches on this problem make assumptions about the underlying  
diffusion models, such as 
the susceptible-infected model ({\si})~\cite{prakash:12:netsleuth,shah:11:culprit}
or the independent-cascade model ({\ic})~\cite{lappas2010finding}, 
or consider additional information, such as 
the exact time of all network interactions~\cite{rozenshtein2016reconstructing}.

\begin{figure}[t]
  \centering
  \includegraphics[width=0.95\linewidth]{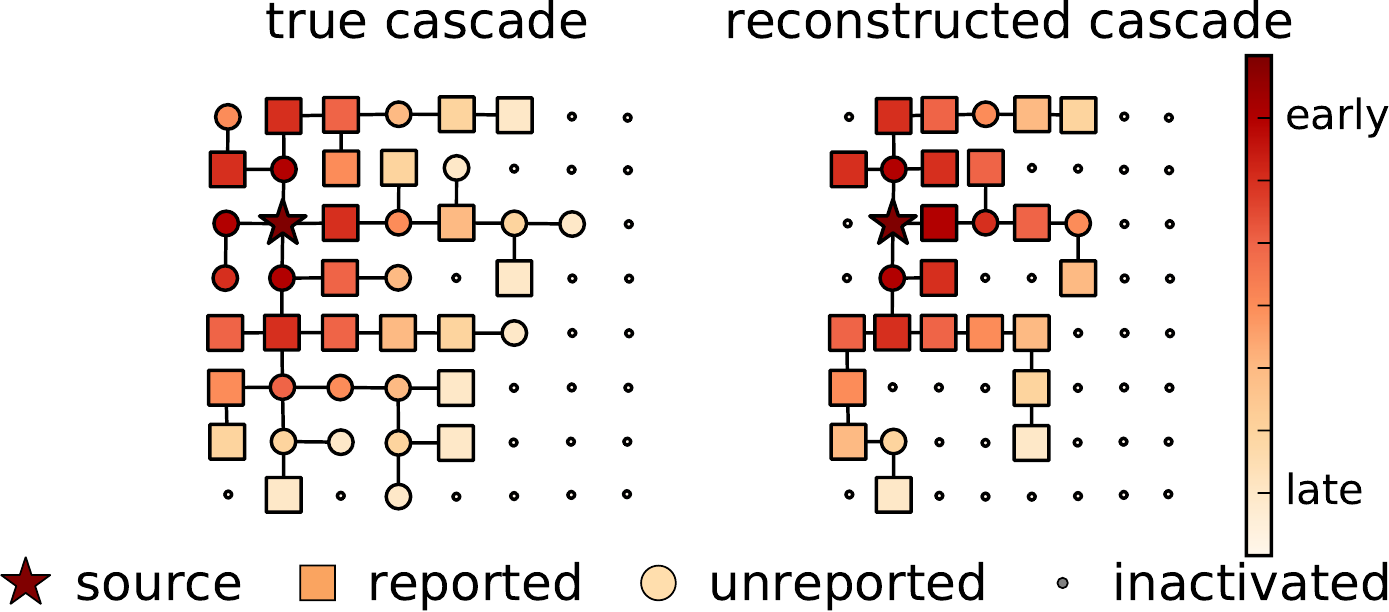}
  \label{figure:intro-example}
  \caption{
    Our method aims at reconstructing the underlying cascade given reported infections (left).
    The reconstructed cascade (right) is parsimonious and respects node infection order induced by infection time.
    Node infection time is indicated by the colorbar. }
\end{figure}

Our approach to the cascade-reconstruction problem 
is to utilize temporal information about the active nodes. 
In particular, we consider that for all observed active nodes
the activation time is known. 
This is a realistic assumption in many settings, 
e.g., it is often easy to determine when a patient got sick. 
However, we still consider the case of partial observations, 
i.e., $\nodes'$ is a subset of all active nodes in the cascade. 

A simple example is illustrated in Figure~\ref{figure:intro-example}.
A {\em ground truth cascade} is depicted in the left side. 
Only a subset of the infected nodes are known, 
and for those nodes we also know their infection time,
which is indicated by the colorbar in the right.
A cascade reconstructed by our algorithms, 
for this problem instance, is shown in the right side. 
The reconstructed cascade is required to parsimonious and to respect
the node infecting order induced by the input infection times.
We note that we do not make any assumption about the diffusion model; 
we only rely on parsimony and ask to find the smallest cascade
tree that is consistent with the observed data. 
Formulated in this manner, our approach is able to find the most important nodes
of the cascade and thus, it gives solutions of high precision.

From the technical viewpoint, 
we formulate the cascade-reconstruction problem as a {\em Steiner-tree} problem.
Given the reported active nodes and their activation timestamps, 
we seek to find a tree that spans all reported nodes, 
while all rooted paths in the tree preserve the order of the observed timestamps. 
We refer to such a tree by {\em order-preserving Steiner tree}
and to the problem we study by \Ourproblem.
This is a novel Steiner-tree problem variant.

For the proposed problem
we develop three approximation algorithms. 
The first algorithm, \algclosure, 
uses the metric closure of the graph induced by the reported active nodes:
it constructs a directed graph, finds a minimum spanning tree on that graph, 
and builds a cascade tree based on the minimum spanning tree.
The second algorithm, \alggreedy, 
is a simplification of \algclosure and it builds a cascade tree directly from the reported active nodes, 
without finding a minimum spanning tree on the metric closure. 
The third algorithm, \algmst, 
builds the cascade tree by a modified breadth-first search (\bfs)
that takes into account the activation timestamps to delay certain search branches.

All three algorithms come with provable approximation guarantees. 
Assuming that the number of reported active nodes is \nrep, 
\algclosure provides a $\bigO(\sqrt{\nrep})$-approximation guarantee
while both \alggreedy and \algmst give a $\nrep$-approximation guarantee.
In practice, when compared to a lower-bound obtained by the standard (non-temporal) Steiner-tree problem, 
all three algorithms give much better solutions than it is suggested by the theoretical guarantee.

In terms of scalability, 
the running time of \algclosure and \alggreedy is $\bigO(\nrep\nedges + \nrep^2)$ and $\bigO(\nrep\nedges)$ respectively, 
where $\nedges$ is the number of edges of the graph \graph. 
On the other hand, \algmst runs in time $\bigO(\nedges + \nrep\log\nrep)$, 
making it an extremely scalable algorithm, 
able to cope with very large graphs and large number of reported active nodes.

In summary, our contributions are as follows. 
\begin{squishlist}
\item
We present a novel formulation of a Steiner-tree problem 
with ordering constraints, 
which models the problem of reconstructing a cascade
from partial observations with temporal information.
\item 
For the proposed Steiner-tree problem
we provide three approximation algorithms. 
The best algorithm in terms of quality
achieves a $\bigO(\sqrt{\nrep})$-approximation guarantee
and has running time $\bigO(\nrep\nedges)$.
The most scalable algorithm 
achieves a $\nrep$-approximation guarantee 
and has running time $\bigO(\nedges + \nrep\log\nrep)$.
\item 
We experimentally evaluate the proposed algorithms
on real-world networks
using cascades simulated with different diffusion models.
Comparison with the baseline Steiner-tree algorithm
that does not use temporal information, 
shows that incorporating timestamps in the problem setting
does not lead to significant loss in solution quality (as measured by the objective function)
while allowing to find other active nodes with better predicted activation time.
\end{squishlist}

The rest of the paper is organized as follows. First, we discuss the related work in Section~\ref{rel}. 
Then we introduce necessary notation and definitions in Section~\ref{section:preliminaries}. 
Section~\ref{pf} presents the problem formulation, followed by Section~\ref{approx} 
where we discuss and analyze our approximation algorithms. 
In Section~\ref{exp}, we evaluate the performance of the proposed algorithms,
and Section~\ref{concl} concludes the paper.

\section{Related work}
\label{rel}

Diffusion processes have been widely studied in general. 
However, the problem of ``reverse engineering'' an epidemic has received relatively little attention.
Shah and Zaman~\cite{shah:11:culprit} formalize the notion of \textit{rumor-centrality} 
to identify the single source node of an epidemic under the susceptible-infected model ({\si}), 
and provide an optimal algorithm for $d$-\textit{regular trees}.
Prakash et al.~\cite{prakash:12:netsleuth} study
the problem of recovering multiple seed nodes under the {\si} model using the 
minimum-description length principle ({\mdl}) ---
i.e., as in this paper a parsimony principle is used.
Lappas et al.~\cite{lappas2010finding} study the problem of 
identifying $k$ seed nodes, or \textit{effectors}, 
in a partially-activated network, 
which is assumed to be in steady-state under the independent-cascade) model ({\ic}). 
Feizi et al.~\cite{feizi2014network} and
Sefer et al.~\cite{sefer:14:diffusion} address the same problem
(identifying $k$ seed nodes) in the case of multiple snapshots of a graph.
Feizi et al.\ consider the {\si} model, while Sefer et al.\ consider 
the susceptible-exposed-infectious-recovered model ({\seir}) model. 
These works consider fully-observed infection footprints and focus only 
on the source-detection problem. 

A recent approach on the cascade-reconstruction problem was proposed
by Rozenshtein et al.~\cite{rozenshtein2016reconstructing}, 
and like this work temporal network information is used. 
However, unlike our setting, 
Rozenshtein et al.\ assume that the complete temporal network is given,
i.e., timestamps are available for all edges. 
%
Thus, the setting is different than the one studied in this paper 
and it makes stronger assumptions about data availability. 

The recent study of Farajtabar et al.~\cite{farajtabar2015back} is one of the closest works to ours. 
However, they consider the problem of identifying a single seed given multiple partially observed cascades. 
and they explicitly assume continuous time diffusion model. 
Another related work is presented by Sundareisan et al.~\cite{sundareisan2015hidden}, 
who simultaneously find the starting points of the epidemic and the missing infections 
given one sample snapshot and
assuming the {\si} model. 
In contrast, our paper addresses the general problem of reconstructing a cascade, 
given several reported active nodes, and corresponding activation timestamps, 
but without assuming any model.

From the theoretical point of view, our problem formulation is a generalization on minimum Steiner-tree problem. This classic \np-complete problem has a folklore 2-approximation algorithm via minimum spanning tree. However, the order constrains make our problem more related to minimum \emph{directed} Steiner tree, which is known to be inapproximable to better than logarithmic factor due to reduction from minimum set cover. 
The best known algorithm was developed by Charikar et al.~\cite{charikar1999approximation}. 
and 
recently improved by Huang~\cite{huang2015minimum}. 
The algorithm constructs the tree by greedy recursion, 
and obtains a guarantee of $\ell(\ell-1)k^{\frac{1}{\ell}}$, 
where $\ell$ is the depth of recursion and $k$ is the number of terminals.
However, the running time is $\bigO(\nnodes^\ell k^\ell)$, and thus, impractical.
\section{Preliminaries}
\label{section:preliminaries}

We consider an undirected graph $\graph=(\nodes, \edges)$ with \nnodes nodes and \nedges edges.
We assume that a \emph{dynamic propagation process} is taking place in the network. 
We use the generic term \emph{active}  
to refer to nodes that have reacted positively during the propagation process, 
e.g., they have infected by the virus, adopted the meme, etc.

We assume that the propagation process starts at some seed node \seed 
and other nodes become active via edges from their active neighbors.
The activation spread can occur according to an unknown model. 
As we will see, our problem definition and algorithms
rely only on a parsimony principle, 
and do not depend on the underlying activity-propagation model.

We only observe a subset of all active nodes in the graph.
For each observed active node $u$
we also obtain the time $t$ when $u$ was activated.
The set of active nodes together with their activation timestamps is 
denoted by $\rep=\{(u,t)\}$, and it is referred to as {\em reported nodes}.
The number of reported active nodes is denoted by $|\rep|=\nrep$.
In many applications we have $\nrep\ll\nnodes$, 
but this is not a required assumption.
The seed \seed does not necessarily belong in
the set of reported active nodes.

The set of active nodes in \rep (i.e., without timestamps) is denoted by $\nodes(\rep)$. 
Similarly, the set of all timestamps in \rep is denoted by $\ts(\rep)$.
We use $\tmin$ to denote the earliest timestamp in $\ts(\rep)$.
If $(u, t) \in \rep$, we write $t(u) = t$. 
The set of all nodes with timestamp $t$ is denoted by $\nodes(t) = \set{u \mid (u, t) \in \rep}$. 

Consider a candidate seed node $r$ and a reported node $(u_1,t_1)\in\rep$. 
A path $\path$ from $r$ to $u_1$ is called \emph{order-respecting path}
if \path does not contain any reported node $(u_2,t_2)\in\rep$, with $t_2 > t_1$.
Note that an order-respecting path can contain two reported nodes with the same timestamp.

\section{Problem formulation}
\label{pf}

Given a graph $\graph=(\nodes, \edges)$ and a set of reported nodes $\rep=\{(u,t)\}$, 
our goal is to reconstruct the most likely cascade that has generated the observed data. 
As mentioned before, we do not make any assumption on the underlying propagation model, 
we only assume that a cascade starts at a seed node (which is unknown) 
and proceeds via graph neighbors.
Motivated by a {\em parsimony consideration}, such as Occam's razor, 
we can formulate the cascade-reconstruction problem 
as finding the {\em smallest} cascade that explains the observed data.
It is clear that such a {\em minimal} cascade
consists of a {\em tree} $\tree$ rooted at the seed node 
and containing all the reported nodes in $\rep$.
Our goal is to infer such a tree $\tree$. 
Furthermore, we should ensure that the reconstructed tree
is {\em consistent} with the observed timestamps in $\rep$.
One natural way to provide a notion of {\em temporal consistency}
is to require that all paths in the reconstructed tree \tree are {\em order-respecting}.
The above discussion motivates our problem definition.

\begin{problem}[\Ourproblem]
\label{problem}
We are given a graph $\graph=(\nodes, \edges)$ and a set of reported nodes 
$\rep=\{(u,t)\}$ with $u\in \nodes$ and $t\in\real$ 
for a relatively small subset of nodes in \nodes.
The goal is to find a seed $\seed\in\nodes$ and 
a tree $\tree$ rooted in $\seed$, such that
\begin{squishlist}	
\item[{\em (}$i${\em )}]
the $\tree$ spans all the reported nodes $\nodes(\rep)$, 
\item[{\em (}$ii${\em )}]
all paths in \tree starting at node $\seed$ are order-respecting, and
\item[{\em (}$iii${\em )}]
the total number of edges in \tree is minimized. 
\end{squishlist}
\end{problem}

Requirements ($i$) and ($ii$) ensure that the reconstructed tree
is consistent with the observed data, 
while requirement ($iii$) quantifies parsimony.

It is clear that for an adversarially-selected subset of reported nodes
it is impossible to reconstruct the ground-truth cascade.
For instance, consider the case that all the nodes in the graph are active,
but the reported nodes are concentrated in a local neighborhood of the graph.
In this case, it is unreasonable to expect from 
any reconstruction algorithm to infer the ground-truth cascade. 

\newpage
On the other hand, if the set of reported nodes is a ``representative'' subset 
of the active nodes (e.g., a uniform sample), 
we expect that it will be possible to reconstruct a tree 
that contains the most {\em salient} and most {\em central} nodes of the cascade. 
In other words, we expect that, compared to all active nodes, 
the set of nodes contained in our reconstructed tree has {\em high precision},
but {\em low recall}.

We now proceed to discuss the algorithmic complexity of Problem~\ref{problem}.
A first observation is
that we are asking to find both a seed node \seed and a tree \tree rooted in~\seed.
A simple way to achieve this
is to consider each node $v\in\nodes$ as a candidate seed, 
find a tree $\tree_v$ rooted in $v$, 
and return the optimal tree (smallest number of edges) 
among all trees found.

Interestingly, it turns out that one does not need to consider all possible
nodes as candidate seeds. 
Instead the optimal tree is the one that is rooted at the reported node with the 
earliest timestamp.
For the following observation recall that $\tmin = \min \ts(\rep)$.

\begin{observation}
\label{observation:root}
The optimal solution to the \Ourproblem\ problem
is a tree rooted at the reported node $(u_0, \tmin)\in \rep$. 
If there is more than one node in the set $\nodes(\tmin)$, 
any of them can be considered as a root for the optimal tree.
\end{observation}

A consequence of Observation~\ref{observation:root}
is that we can restrict our attention to problem \routedproblem, 
a variant of \Ourproblem where the seed \seed is given as input.

Even though Observation~\ref{observation:root} gives a useful and practical optimization, 
it does not change the computational complexity of the problem we consider.

\begin{proposition}
Problem \Ourproblem is \nphard.
\end{proposition}

\begin{corollary}
Problem \routedproblem\ is \nphard.
\end{corollary}

\noindent
All proofs are provided in the appendix.

\section{Approximation algorithms}
\label{approx}

The Steiner tree problem, and many of its variants, 
have been studied extensively in the literature. 
Different approximation algorithms are available 
depending on the exact setting and problem variant. 

In the standard Steiner tree problem~\cite{williamson2011design}
we are given an undirected weighted graph $\graph=(\nodes, \edges)$, 
and a set of terminal nodes $\terminals\subseteq\nodes$,
and the goal is to find tree that spans all terminal nodes
and whose total edge weight is minimized. 
For this problem there are several approximation algorithms
with a constant-factor approximation guarantee, 
such an algorithm that uses the metric closure 
of the graph induced by the terminals~\cite{williamson2011design},
or the primal-dual method~\cite{goemans1995}.

The Steiner tree problem on directed graphs is considerably more difficult; 
the best known algorithm, proposed by Charikar et al.~\cite{charikar1999approximation}, 
gives an approximation guarantee
$\ell(\ell-1)k^{\frac{1}{\ell}}$, for recursion parameter $\ell>2$, 
and has running time $\bigO(\nnodes^\ell k^\ell)$,
while it is known that the problem
cannot be approximated with a factor better than 
$c\log|X|$, unless $\mathbf{P}=\mathbf{NP}$.

The \Ourproblem problem, proposed in this paper, 
is a novel Steiner-tree variant, 
and thus, new approximation algorithms need to be devised.
The problem formulation assumes an undirected graph,
but ordering constraints are required
to ensure that the reconstructed cascade is consistent with the observed timestamps.
This makes the problem significantly different than 
existing formulations on either undirected or directed graphs.

\subsection{Algorithm based on metric closure}

Our first algorithm for the \Ourproblem problem
uses the metric closure of the graph induced by the terminals,
and it is an adaptation of the algorithm 
for the standard Steiner-tree problem on undirected graphs.
We call this algorithm \algclosure.

Some additional notation is needed to present the \algclosure algorithm.
Consider an instance of \Ourproblem, 
i.e., a graph $\graph=(\nodes,\edges)$ and a set of reported nodes $\rep=\{(u,t)\}$.
Given two reported nodes $u$ and $v$ 
we define the \emph{excluding shortest path} from $u$ to $v$ 
to be the shortest path from $u$ to $v$ in $G(V\setminus \nodes(\rep),E)$, 
that is, the shortest path that does not use any other reported nodes. 
If there are multiple shortest paths, we select one arbitrarily.
The excluding shortest path from $u$ to $v$ is denoted by $\exclpath(u,v)$, 
and its length by $\excldist(u, v) = \abs{\exclpath(u, v)}$.

The \algclosure algorithm
constructs a weighted directed graph $H$ among the reported nodes. 
Edge directions in $H$ are from terminals of earlier timestamps
to terminals of later timestamps, 
and edge weights correspond to excluding shortest-path lengths.
Next, \algclosure constructs a directed minimum spanning tree on $H$ rooted in \seed.
Then the algorithm reconstructs a cascade on the original graph \graph
by starting from seed \seed and 
processing the terminals $\nodes(\rep)$ in chronological order:
each new terminal $u$ 
is added to the currently-reconstructed cascade
via the best path that goes through one of its ancestors in the minimum spanning tree on $H$.
The pseudocode from \algclosure is given in Algorithm~\ref{alg:arbor}.
\begin{algorithm2e}[t]
  \caption{\algclosure}
  \label{alg:arbor}
  \KwInput{$\graph=(\nodes,\edges)$, $\rep=\{(u,t)\}$, root $\seed\in\nodes(\tmin)$}
  \KwOutput{Order-respecting 
    \tree rooted in \seed that spans all nodes in $\nodes(\rep)$.}	$Q \define \nodes(\rep)$\;
  Construct weighted directed graph $H=(Q,X)$ with $X=\{(u,v): t(u)\leq t(v)\}$ and edge weights 
  $w(u,v)= \excldist(u,v)$\;
  $A \define$  minimum directed spanning tree of $H$ rooted in \seed\;
  $\tree \define (\{\seed\},\emptyset)$  \tcp*{initially, \tree contains seed \seed and no edge}
  \ForEach{$u\in \nodes(\rep)$ in chronological order}
  {
    \nosemic find $w$ so that\;
    \dosemic \pushline
    $w$ is in $T$\;
    $w$ is the closest ancestor of $u$ in $A$\;
    \popline add path from $u$ to $w$ in $A$ to $\tree$\;
  }	
  \Return {\tree}
\end{algorithm2e}

We first prove that \algclosure returns a feasible solution. 

\begin{proposition}
\algclosure returns an order-respect\-ing Steiner tree, which spans $\rep$.
\end{proposition}

Next we show that \algclosure provides a
$\bigO(\sqrt{\nrep})$-ap\-pro\-xi\-ma\-tion guarantee ---
recall that  $\nrep=|\nodes(\rep)|$ is the number of terminals (reported nodes).

\begin{proposition}
\label{prop:closureapprox}
The \algclosure algorithm provides
approximation guarantee $2(1+\sqrt{3/2})\sqrt{\nrep - 1}$ for problem \Ourproblem.
\end{proposition}

The running time of \algclosure algorithm is as follows. 
Computing shortest-path distances 
to construct the graph $H$ requires time $\bigO(\nedges\nrep)$, 
while finding the minimum directed spanning tree on $H$ requires time $\bigO(\nrep^2)$.
In the second phase of \algclosure, 
constructing the Steiner tree \tree requires time $\bigO(\nnodes\nrep)$. 
Thus, the overall running time of \algclosure algorithm is $\bigO(\nrep\nedges + \nrep^2)$.

\subsection{Greedy algorithm}

Our second algorithm, \alggreedy, is a simpler variant of \algclosure. 
\alggreedy avoids the first step of computing 
a minimum spanning tree, 
and instead, it reconstructs the cascade by adding paths to the terminals
in chronological order---i.e., similar to the second phase of \algclosure. 

Given a terminal node $u\in\nodes(\rep)$ with timestamp $t(u)$
and a graph node $v\in\nodes$
we define the \emph{extended excluding shortest path} 
from $v$ to terminal $u$ 
to be the shortest path from $v$ to $u$ in \graph, 
which may include only terminals $w$ with the same timestamp as $u$.
The extended excluding shortest path from $v$ to $u$ is denoted by $\extexclpath(v,u)$,
and its length by $\extexcldist(v,u)$.

The \alggreedy algorithm starts by adding to the cascade only the seed node, 
i.e., $\tree \define (\{\seed\},\emptyset)$.
Then \alggreedy processes the reported nodes in chronological order.
For each reported node$u$, \alggreedy finds
the {\em shortest} extended excluding shortest path $\extexclpath(v,u)$,
over all nodes $v$ that are currently included in the cascade \tree, 
and add this path in \tree.
Pseudocode for \alggreedy is given in Algorithm~\ref{alg:greedy}.

It is clear that \alggreedy returns a feasible solution.
%
Additionally, we can show that \alggreedy provides an appro\-xi\-ma\-tion guarantee
for \Ourproblem, 
albeit a weaker bound than the one obtained by \algclosure. 

\begin{algorithm2e}[t]
	\caption{\alggreedy}
	\label{alg:greedy}
	\KwInput{$\graph=(\nodes,\edges)$, $\rep=\{(u,t)\}$, root $\seed\in \nodes$}
	\KwOutput{Order-respecting 
	\tree rooted in \seed that spans all nodes in $\nodes(\rep)$.}
	$\tree \define (\{\seed\},\emptyset)$\;
	\ForEach{$u\in \nodes(\rep)$ in chronological order\\ \tcp{ties are broken arbitrarily}}
	{
		$v \define \arg\min_{z\in\tree} \extexcldist(z,u)$\;
		$\tree \define \tree \cup \extexclpath(v,u)$\;
	}
		
	\Return {\tree}
\end{algorithm2e}

\begin{proposition}
\label{prop:greedyapprox}
Algorithm \alggreedy yields a $\nrep$-approximation guarantee for the \Ourproblem problem.
\end{proposition}

The running time of \alggreedy is similar to that of \algclosure:
processing each reported node requires a \bfs computation, 
so the overall running time is $\bigO(\nrep\nedges)$.

\subsection{Delayed BFS algorithm}

Both previous algorithms, \algclosure and \alggreedy, 
perform $k$ operations that are equivalent to \bfs, 
and thus their running time is $\bigO(\nrep\nedges)$.
In cases that there are many reported nodes, i.e., when $\nrep$ is large, 
algorithms \algclosure and \alggreedy are not scalable to large graphs.

To address this challenge we propose a third algorithm, \algmst,
which, like \alggreedy, provides a $\nrep$-approximation guarantee,
but is more efficient.
The main idea of \algmst is to perform a single \bfs starting from the root \seed.
Whenever a terminal  $u$ is encountered,
\algmst checks whether all terminals with timestamp smaller than $t(u)$ have been visited.
If they have been visited, the \bfs continues.
If not, the \bfs ``below'' node $u$ is delayed 
until all terminals with timestamp smaller than $t(u)$ have been visited.

Pseudocode for \algmst is given in Algorithm~\ref{alg:dbfs}.
The variable $Q$ represents the \bfs queue 
while $D$ represents a sorted array with the terminals
at which {\bfs} has been delayed; 
the terminals in $D$ are sorted in chronological order.
The variable $t_{\mathit{cur}}$ keeps the 
smallest timestamp of the terminals that have not been processed yet.
The array $c$ keeps for each timestamp $t$ 
the number of terminals with that timestamp that have not been processed yet.
The counter for $c[t]$ is decreased whenever we delete from $D$ a terminal with timestamp~$t$.
Thus, at each step of the algorithm 
the minimum time $t_{\mathit{cur}}$ can be computed as the smallest timestamp~$t$
for which $c[t]\not=0$ in amortized constant time.

\begin{algorithm2e}[t]
    \caption{\algmst}
    \label{alg:dbfs}
    \KwInput{$\graph=(\nodes,\edges)$, $\rep=\{(u,t)\}$, root \seed}
    \KwOutput{Order-respecting 
	\tree rooted in \seed that spans all nodes in $\nodes(\rep)$.}
	$Q \define \set{\seed}$\;
	$D \define \emptyset$\;
	$\tree \define \emptyset$\;
	$c[t] \define \abs{\set{v \mid (v, t) \in \rep}}$\;
	$c[t(\seed)] \define c[t(\seed)] - 1$\;
	$t_{\mathit{cur}} \define \min \set{t \mid c[t] \neq 0}$\;
	
	\While {$Q \neq \emptyset$} {
		$v \define $ top element in $Q$; delete $v$ from $Q$\;
		\uIf {$v \notin \nodes(\rep)$} {
			\ForEach {$(v, w) \in E$ with unmarked $w$} {
				add $(v, w)$ to $\tree$;
				add $w$ to $Q$; mark $w$\;
			}
		}
		\Else {
			add $v$ to $D$\;
		}
		\ForEach {$v$ in $D$ in chronological order} {
			\lIf {$t(v) \neq t_{\mathit{cur}}$}  {
				\Break
			}
			\ForEach {$(v, w) \in E$ with unmarked $w$} {
				add $(v, w)$ to $\tree$;
				add $w$ to $Q$; mark $w$\;
			}
			$c[t(v)] \define c[t(v)] - 1$\;
			$t_{\mathit{cur}} \define \min \set{t \mid c[t] \neq 0}$\;
		}
	prune branches from $\tree$ that do no have terminal leaf\;
	\Return \tree\;
	}
\end{algorithm2e}

As mentioned before, 
algorithm \algmst provides a provable approximation guarantee, 
similar to the one of \alggreedy.

\begin{proposition}
Algorithm \algmst yields a $\nrep$-appro\-xi\-ma\-tion guarantee for the \Ourproblem problem.
\end{proposition}

The running time of \algmst\ is $\bigO(\nedges + \nrep\log\nrep)$.
The $\bigO(\nedges)$ part is due to the \bfs on the graph 
(the delayed order does not impact the running time)
while the $\bigO(\nrep\log\nrep)$ is due to maintaining the arrays $D$ and~$c$.
Thus, \algmst\ is an algorithm with excellent sca\-la\-bi\-lity 
and can be used for very large graphs.

\section{Experimental results}
\newlength{\figwidth}
\label{exp}


\spara{Datasets:}
We experiment on real-world graphs with both simulated and real cascades.
We use the following real-world graphs from SNAP:\footnote{http://snap.stanford.edu/data/index.html}
($i$)
\email: email data from an European research institution.
There is an edge $(u, v)$ if person $u$ has sent at least one email at person $v$.
The graph has 986 nodes and 25\,552 edges;
($ii$)
\grqc: General Relativity and Quantum Cosmology collaboration network from e-print arXiv.
The graph has 4\,158 nodes and 13\,428 edges; 
($iii$) \arx: High Energy Physics collaboration network from the arXiv, 
consisting of 8\,638 nodes and 24\,827 edges; and
($iv$) \fb: ``friends lists'' from Facebook. 
The graph has 4039 nodes and 88234 edges. 

\spara{Cascading models:}
We simulate cascades using four types of models. 
($i$) susceptible infected (\si);
($ii$) independent cascade (\ic);
($iii$) continuous-time diffusion process (\ct)\cite{rodriguez2016influence};
($iv$) shortest path (\spath), in which contagion propagates by shortest paths. 
For \si, infection probability is set to 0.5.
For \ct, infection time is globally distributed by exponential distribution with $\beta=1$.
For \si, \ct and \spath, we continue the cascade until at least half of the nodes are activated. 
For \ic, activation probability is tuned network-wise to activate on expecation half of the nodes.
For \si, \ic, and \spath transmission delay is one time unit. 

\spara{Real cascades:}
We use the \digg dataset.\footnote{https://www.isi.edu/~lerman/downloads/digg2009.html}
The underlying graph has 279\,631 nodes and 1\,548\,131 edges.
Each story corresponds to a cascade.
For most cascades the activated nodes do not form a connected component; 
in such cases we extract the largest connected component.
We experimented on 18 large cascades (average size 1\,965). 

\spara{Methods:}
We compare the following four methods:
($i$) \algone in Algorithm~\ref{alg:arbor};
($ii$) \alggreedy in Algorithm~\ref{alg:greedy};
($iii$) \tbfs in Algorithm~\ref{alg:dbfs}.
We also consider the standard Steiner-tree problem, 
where no temporal information is used.
The resulting algorithm, \baseline, uses the {\sc\large mst}-based technique~\cite{vazirani2013approximation}.
For all methods, earliest reported activation is selected the root according to Observation~\ref{observation:root}. 

\spara{Measures:}
To evaluate the performance of each method, we compare:
($i$) objective function value defined in Problem~\ref{problem};
($ii$) precision and recall of the set of activated nodes inferred by the tree
with respect to the actual activated nodes;
($iii$) order accuracy: an edge $u \rightarrow v$ in the tree is \textit{correct} if it respects the true infection order, $t(u) \le t(v)$. 
Order accuracy is the fraction of correct edges in the predicted tree. 

For all simulated cascades, we experiment with reporting probabilities with exponential increase,
$q=\set{0.001 \times 2^i \mid i=0, \ldots, 8}$.
For real cascades, we experimented with $q=\set{0.001 \times 2^i \mid i=1, \ldots, 5}$.

For simulated cascades, measurements are averaged over 100 runs for each experiment setting, 
while for real cascades, we average over 8 runs. 

\spara{Objective function:}
Figure~\ref{fig:obj} shows the average tree size
with respect to a the fraction of reported nodes. 
We observe that \tbfs produces larger trees than the other methods 
because it does not explicitly minimize tree size.
We expect \baseline to give the smallest trees
as it does not impose any ordering constraint. 
However, we observe that  
all methods give comparable sizes.

\begin{figure}[t]
  \setlength{\figwidth}{4.0cm}
  \setlength{\tabcolsep}{0pt}
  \begin{tabular}{rr}
    \includegraphics[width=\figwidth]{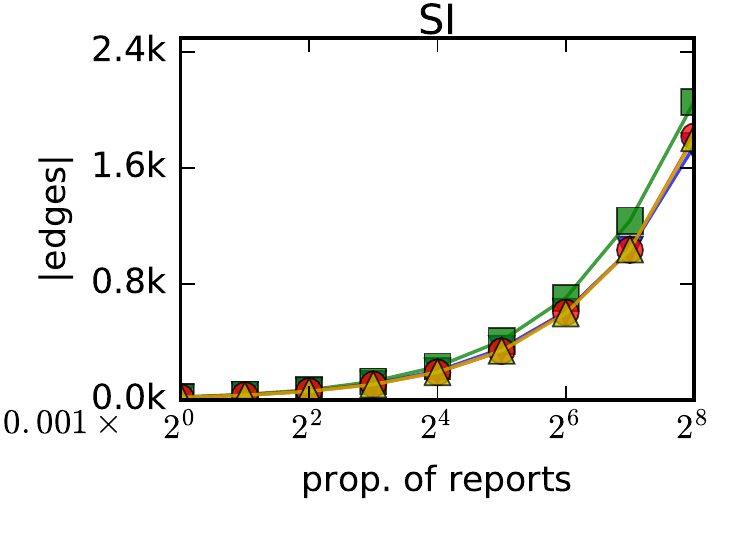} & 
                                                                                \includegraphics[width=\figwidth]{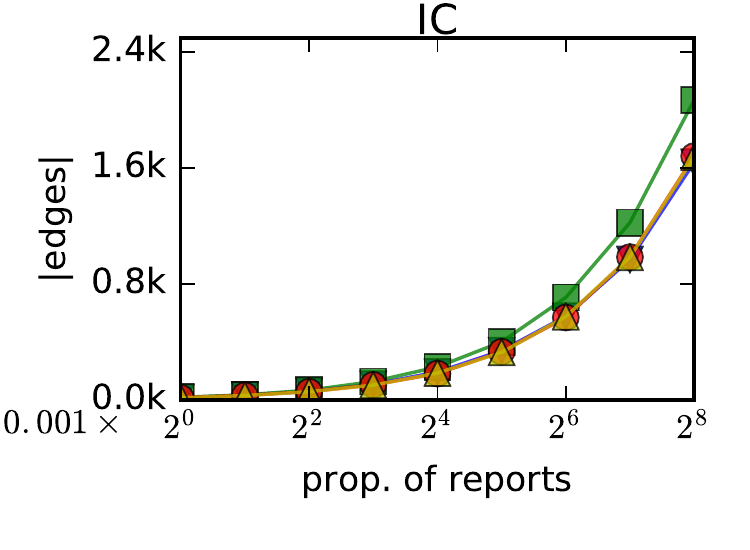} \\
   \multicolumn{2}{r}{\includegraphics[width=2\figwidth]{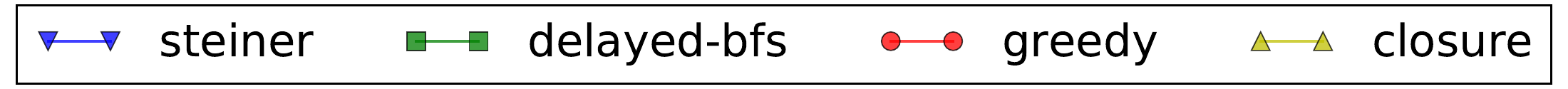}}
  \end{tabular}  

  \caption{Tree size (objective function)
    produced by different methods in graph \arx,
    under \si (left figure) and \ic (right figure)}
  \label{fig:obj}
\end{figure}

\spara{Precision and recall:}
We next demonstrate our methods' adaptability on different graphs and models with respect to node precision and recall.
Figure~\ref{fig:precision_recall_datasets} varies the graphs while fixing the cascade model, meanwhile Figure~\ref{fig:precision_recall_models} does the other way around.

In both Figure~\ref{fig:precision_recall_datasets} and Figure~\ref{fig:precision_recall_models}, 
we observe that \textit{all} four methods achieve node precision larger than 0.8. 
For most settings, node precision is close to 1.
For \ic model in Figure~\ref{fig:precision_recall_models}, precision drops slightly and \tbfs performs worse than the other three.
Note that even though \baseline does not explicitly cope with infection order, it still achieves equivalent performance. 
This demonstrates that the parsimony consideration in \Ourproblem is reasonable with respect to achieving high node precision. 

For node recall, we observe the following in both figures. 
First, \arx and \grqc,  node recall grows linearly as proportion of reports.
However, for \fb and \enron, node recall grows slower compared to the other two.
The reason is \fb and \enron are graphs with larger density, in which it generally takes \textit{fewer} Steiner nodes to construct a Steiner tree. 

Second, \tbfs tends to achieve higher recall than
the other three 
because it does not explicitly minimize the tree size, therefore it captures more nodes. 
Third, \algone, \alggreedy and \baseline tend to have similar recall.

\begin{figure*}[t]
  \setlength{\figwidth}{3.5cm}\setlength{\tabcolsep}{0pt}
  \centering
  \begin{tabular}{rrrr}
  \includegraphics[width=\figwidth]{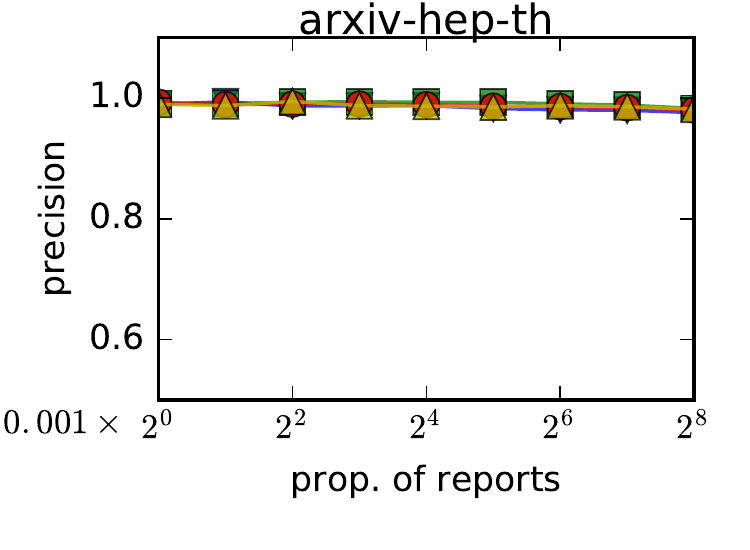}&
  \includegraphics[width=\figwidth]{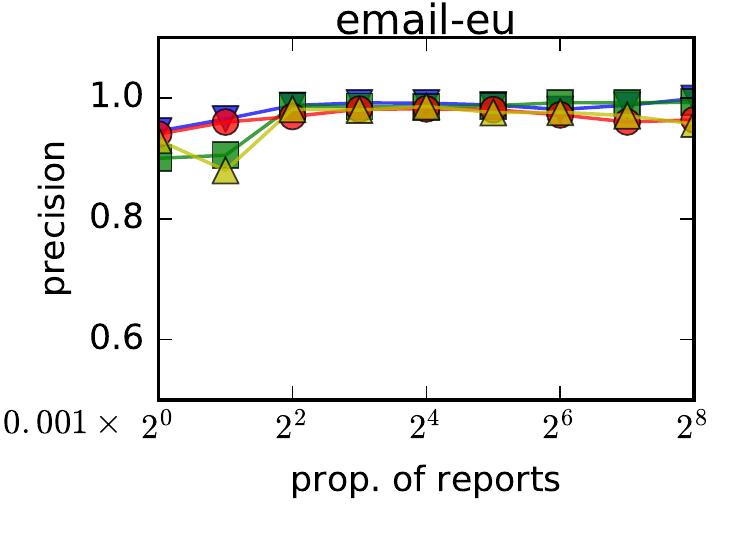}&
  \includegraphics[width=\figwidth]{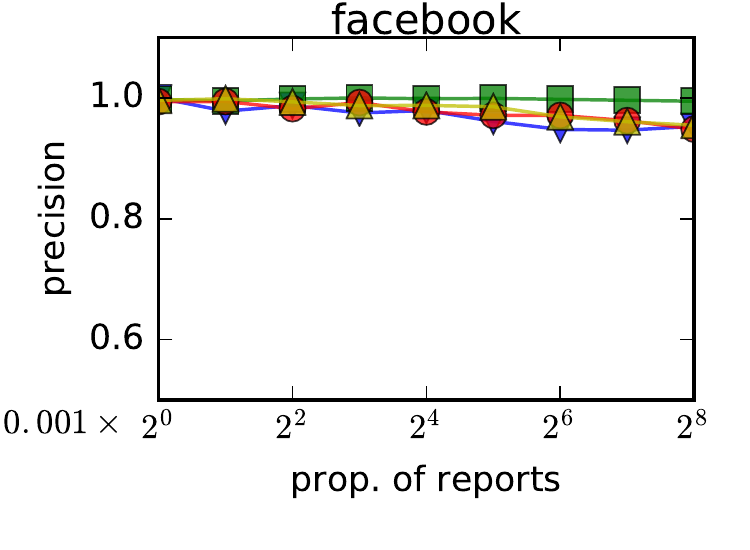}&
  \includegraphics[width=\figwidth]{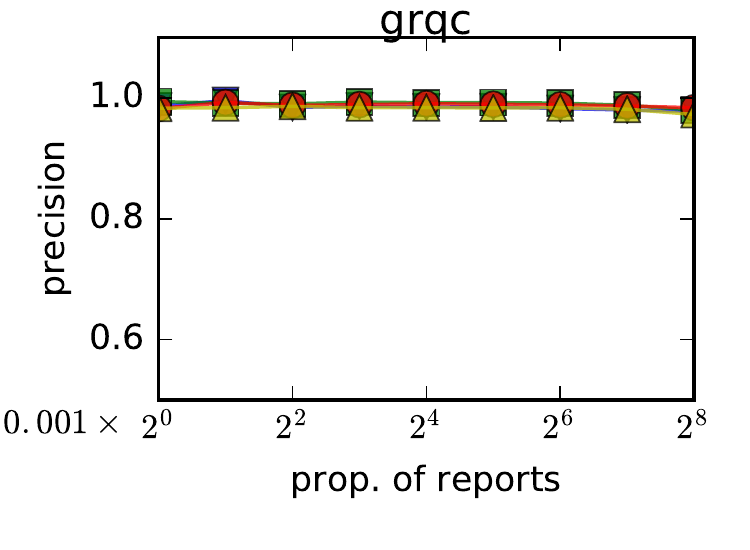}
    \\
  \includegraphics[width=\figwidth]{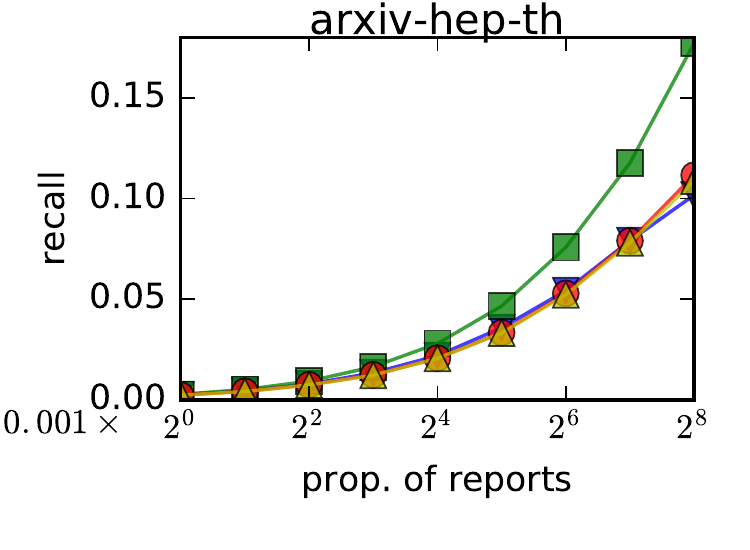}&
  \includegraphics[width=\figwidth]{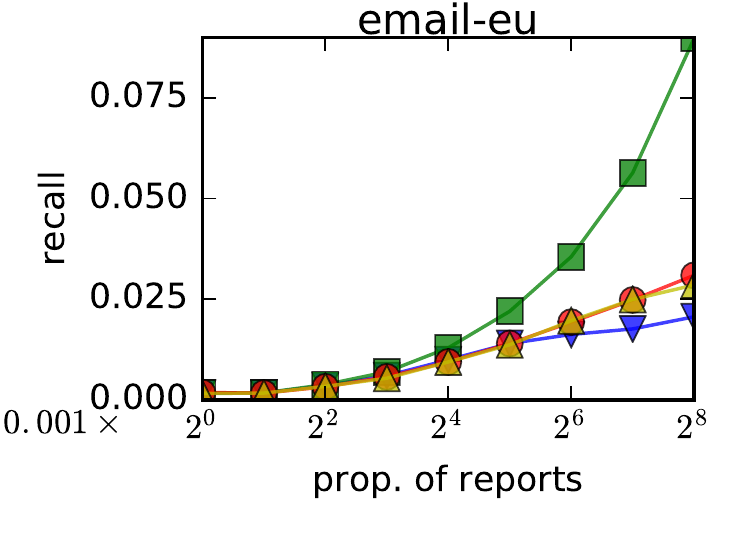}&
  \includegraphics[width=\figwidth]{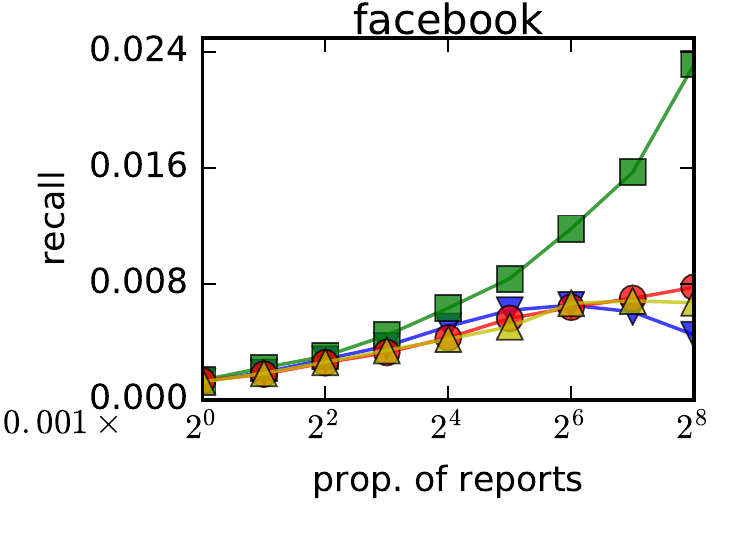}&
  \includegraphics[width=\figwidth]{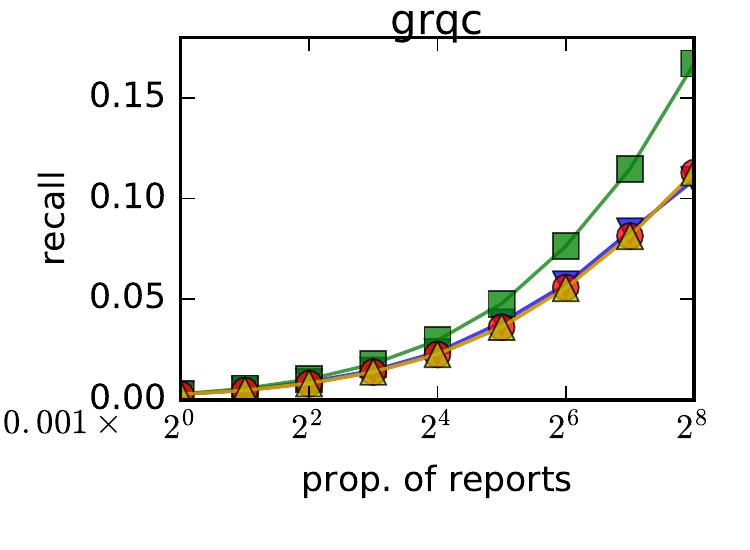}  \\    
  \multicolumn{4}{r}{\includegraphics[width=2.5\figwidth]{figs/si-by_datasets/legend.pdf}}
  \end{tabular}
  \caption{Node precision (upper row) and recall (lower row)
    for different graphs under \si model.}
  \label{fig:precision_recall_datasets}
\end{figure*}

\begin{figure*}[t]
  \setlength{\figwidth}{3.5cm}\setlength{\tabcolsep}{0pt}
  \centering
  \begin{tabular}{rrrr}
\includegraphics[width=\figwidth]{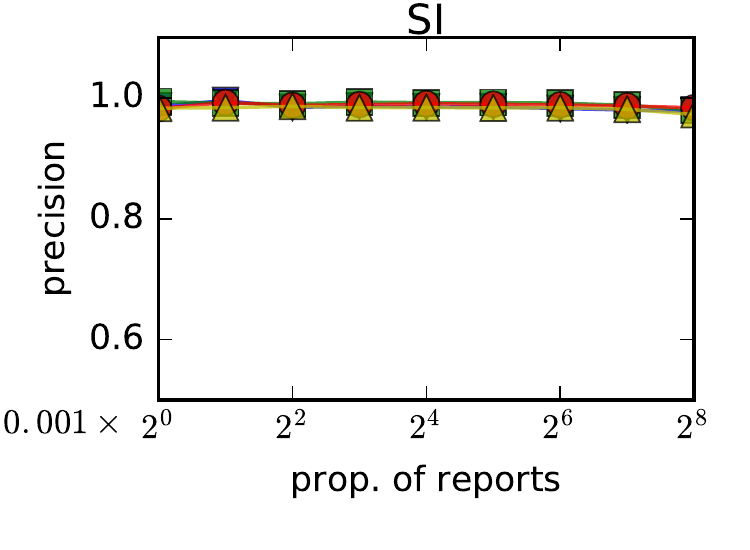}&
\includegraphics[width=\figwidth]{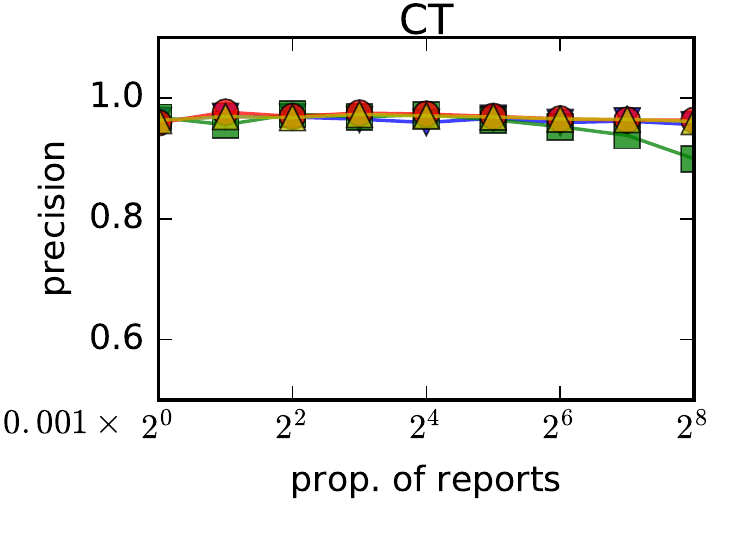}&
\includegraphics[width=\figwidth]{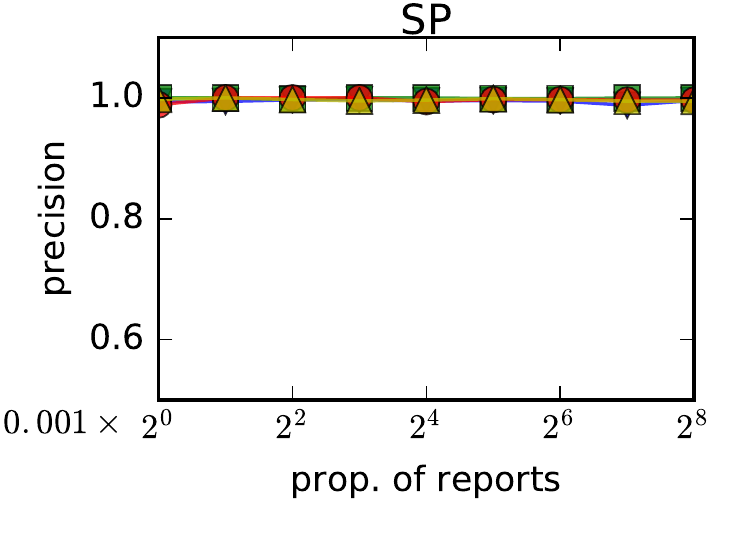}&
\includegraphics[width=\figwidth]{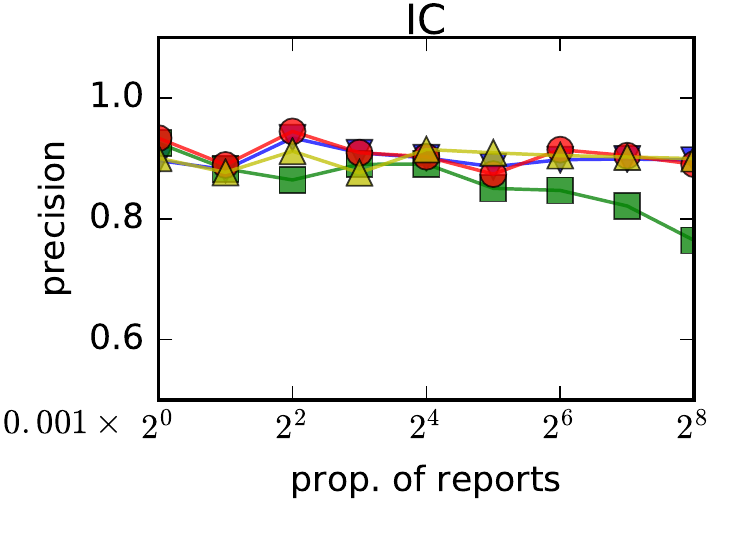}\\
\includegraphics[width=\figwidth]{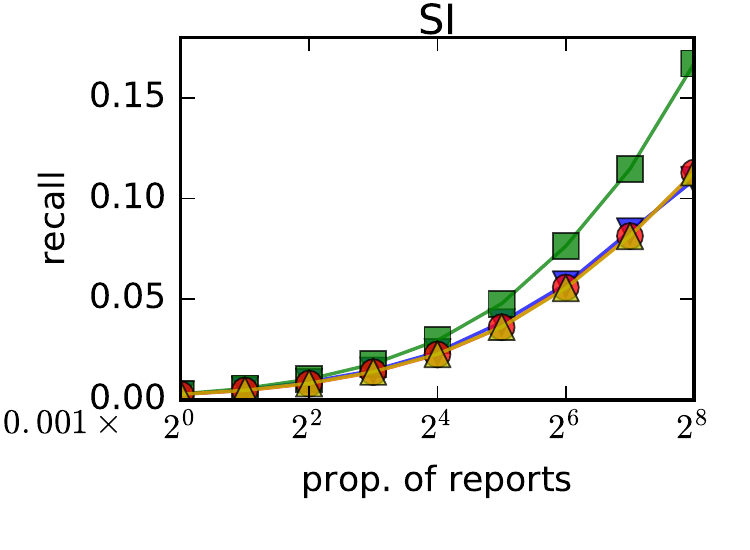}&
\includegraphics[width=\figwidth]{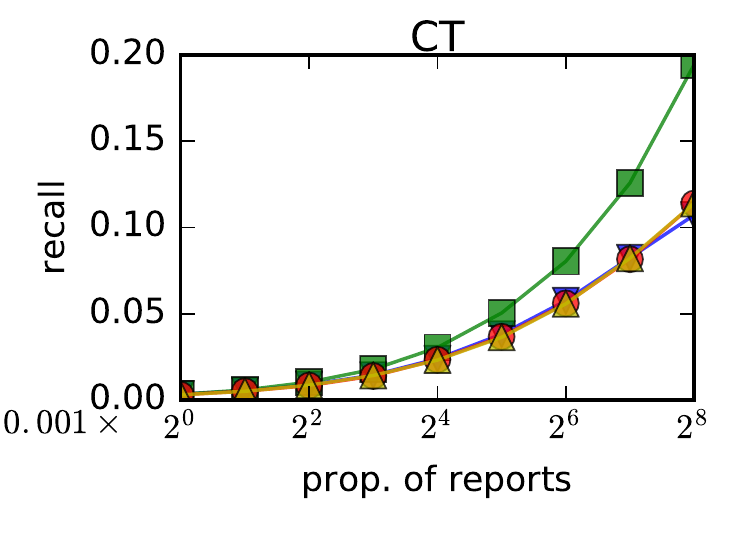}&
\includegraphics[width=\figwidth]{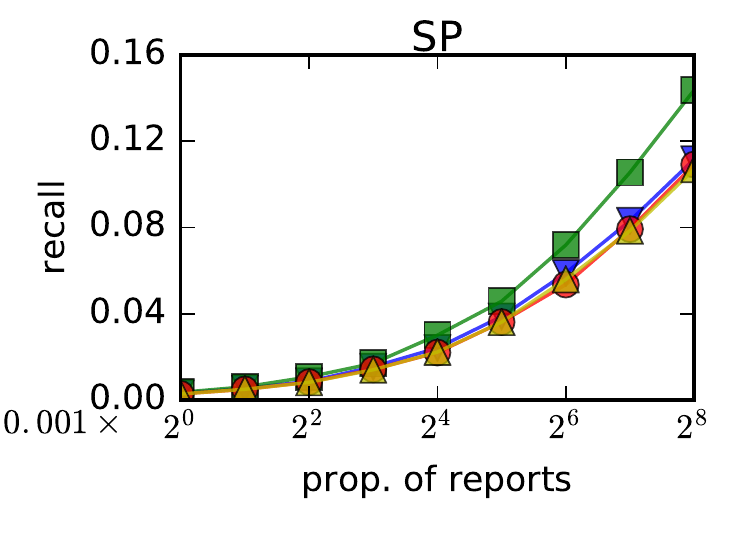}&
\includegraphics[width=\figwidth]{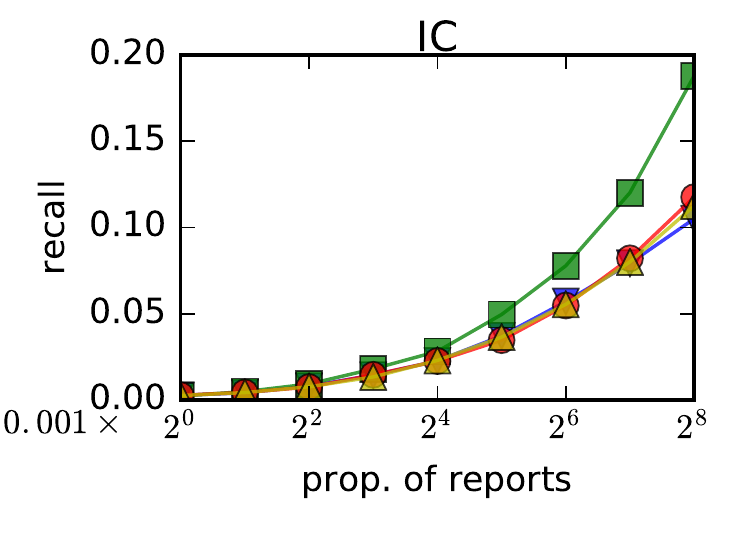} \\    
  \multicolumn{4}{r}{\includegraphics[width=2.5\figwidth]{figs/si-by_datasets/legend.pdf}}
  \end{tabular}
  \caption{Node precision (upper row) and recall (lower row)
    for different models on \grqc.}
  \label{fig:precision_recall_models}
\end{figure*}

\spara{Order accuracy:}
Next, we report order accuracy.
In general, \algone, \alggreedy and \tbfs perform better than \baseline under all cascade models and all graphs.
This advantage is demonstrated in Figure~\ref{fig:order-accuracy}, in which the upper row fixes the cascade model while the lower row fixes the underlying graph. 
This is expected because \algone and \tbfs explicitly construct tree that respect the infection order.

In addition, methods that models order explicitly improve their order accuracy as proportion of reports increase.
However, this is not always true for \baseline.
For example, its order accuracy deteriorates for \fb and \email under \ct in Figure~\ref{fig:order-accuracy}.

\begin{figure*}[t]
  \setlength{\figwidth}{3.5cm}\setlength{\tabcolsep}{0pt}
  \centering
  \begin{tabular}{rrrr}
  \includegraphics[width=\figwidth]{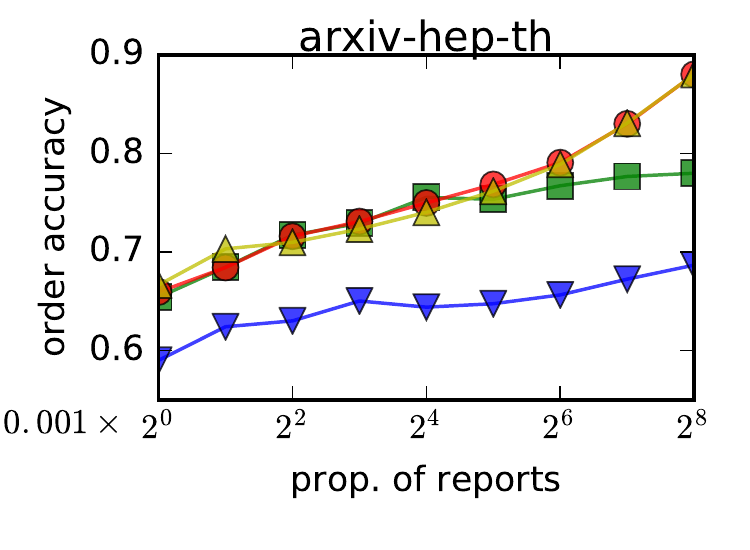}&
  \includegraphics[width=\figwidth]{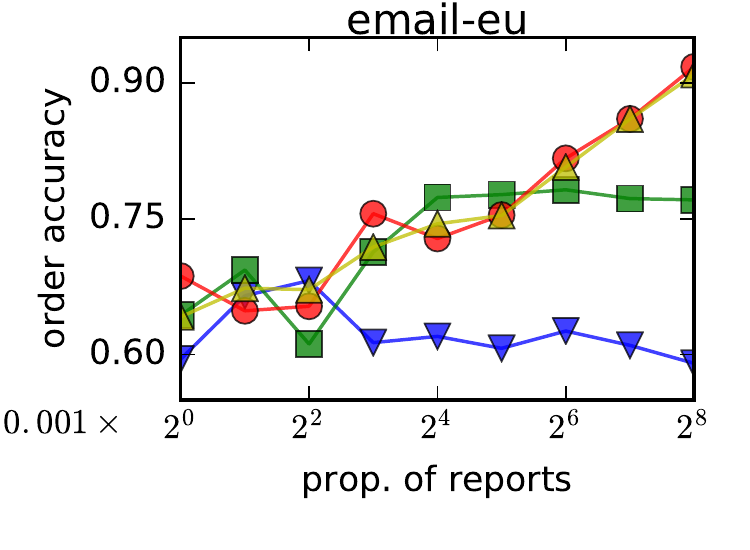}&
  \includegraphics[width=\figwidth]{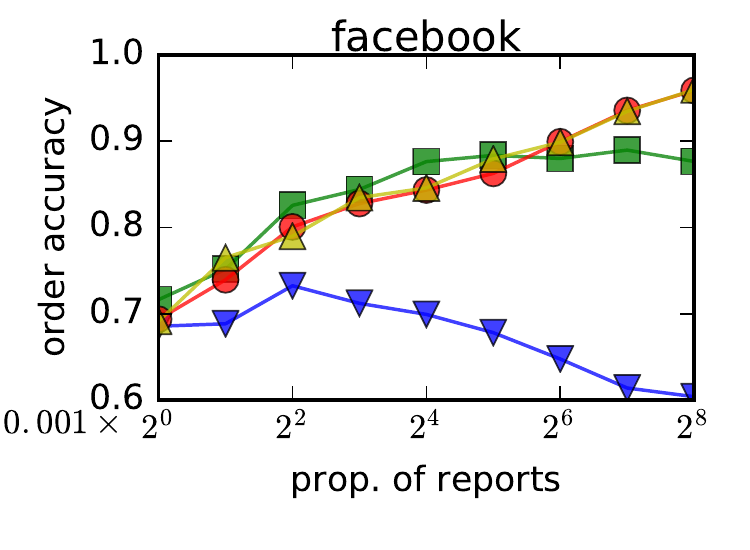}&
  \includegraphics[width=\figwidth]{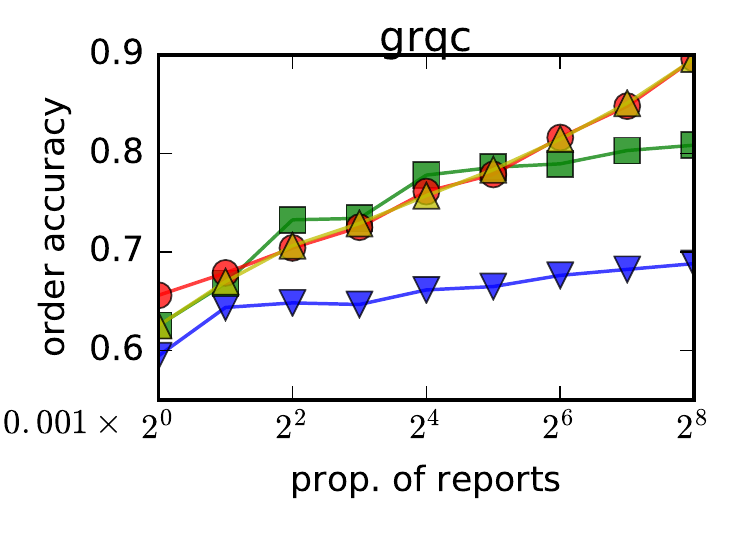} \\

  \includegraphics[width=\figwidth]{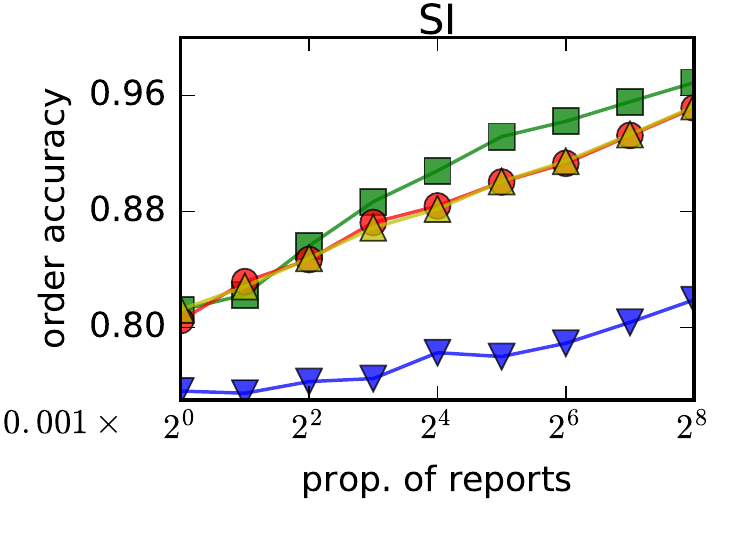}&
  \includegraphics[width=\figwidth]{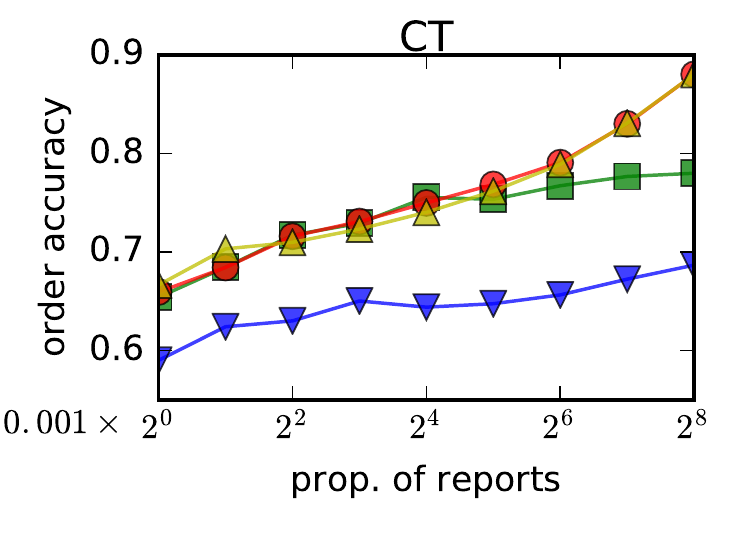}&
  \includegraphics[width=\figwidth]{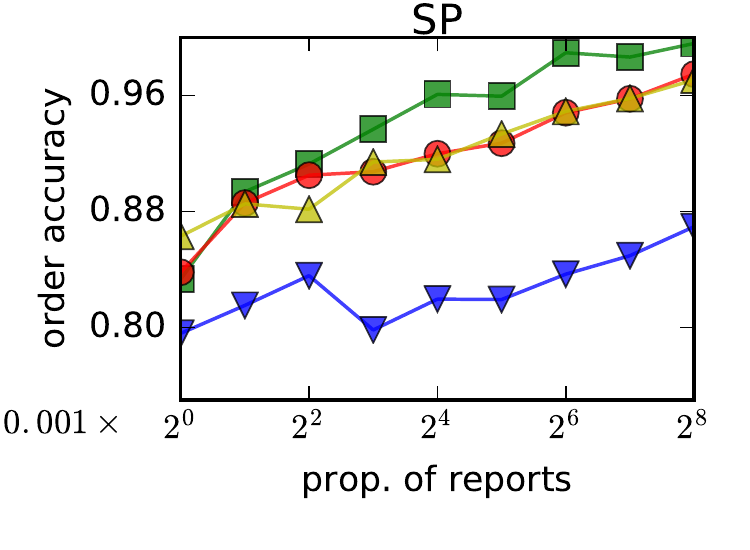}&
  \includegraphics[width=\figwidth]{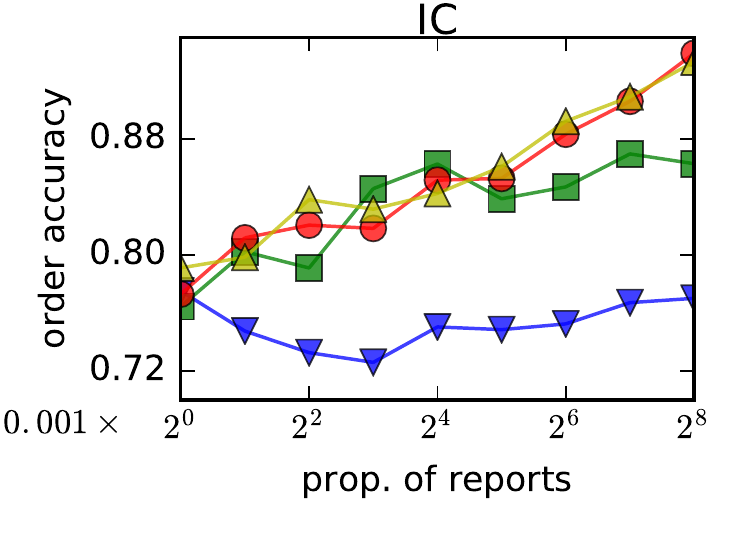}
  \\
  \multicolumn{4}{r}{\includegraphics[width=2.5\figwidth]{figs/si-by_datasets/legend.pdf}}
  \end{tabular}
  \caption{Order accuracy across datasets and cascade models. 
    \textbf{Upper row}: under \ct across all datasets. 
    \textbf{Lower row}: in graph \arx under all models }
  \label{fig:order-accuracy}
\end{figure*}

\spara{Real cascades:}
Performance measure on real cascades is given in Figure~\ref{fig:digg} for both large and small cascades. 
For most the measures, they demonstrate similar behavior with respect to that on synthetic cascades.
However, one noticeable difference is that node precision drops significantly.
The reasons can be two-fold: 1) the infected nodes are tightly connected with each other as well as other uninfected nodes.
In some cases, the uninfected nodes serve as good Steiner nodes
2) parsimony assumption does not hold for certain real cascades. 

\begin{figure*}[t]
  \setlength{\figwidth}{3.5cm}\setlength{\tabcolsep}{0pt}
  \centering
  \begin{tabular}{rrrr}
  \includegraphics[width=\figwidth]{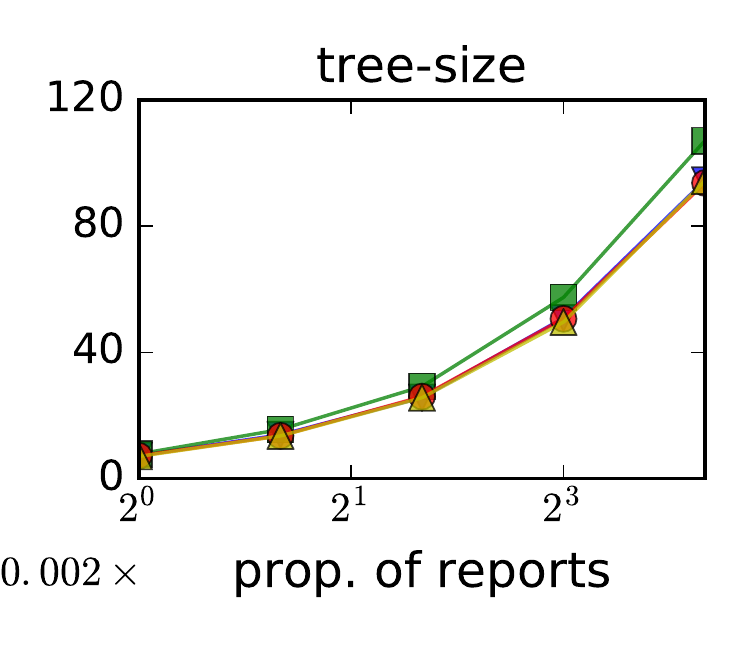}&
  \includegraphics[width=\figwidth]{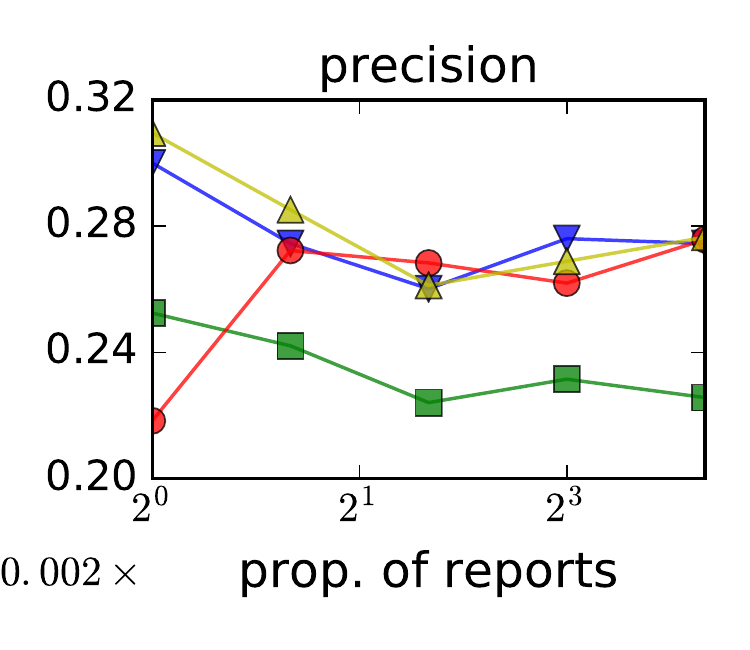}&
  \includegraphics[width=\figwidth]{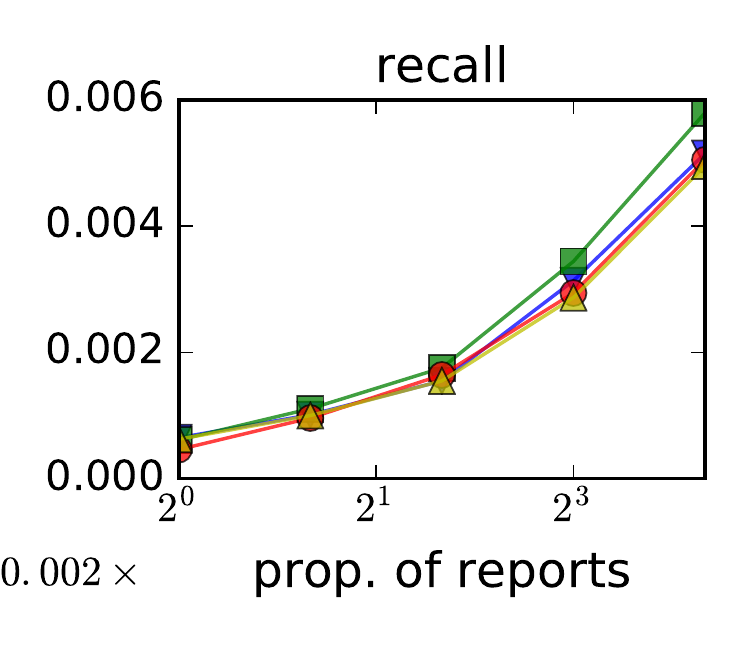}&
  \includegraphics[width=\figwidth]{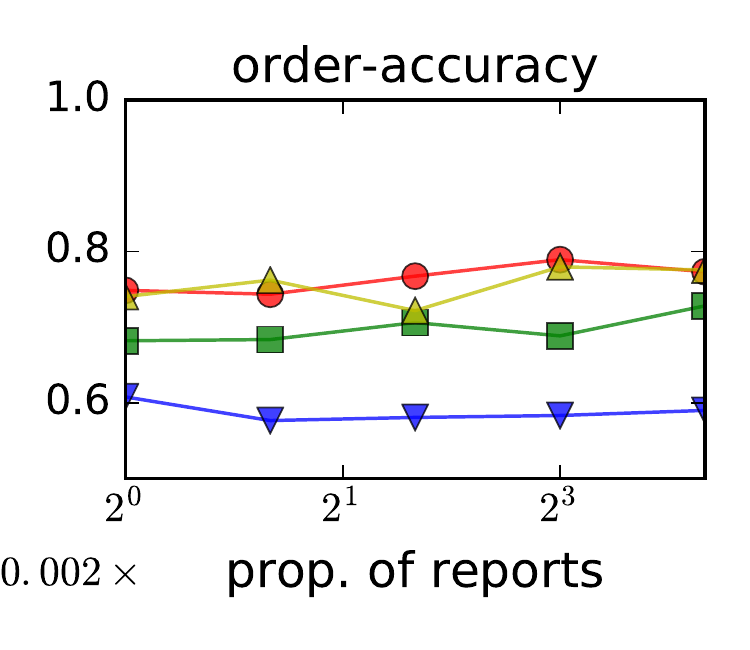}
    \\
    \multicolumn{4}{r}{\includegraphics[width=2.5\figwidth]{figs/si-by_datasets/legend.pdf}}
  \end{tabular}
  \caption{Performance measure on real cascades in \digg.
  }
  \label{fig:digg}
\end{figure*}

\spara{Scalability:}
Last, we evaluate the scalability of \alggreedy. 
We conduct experiments on a 2.5\,GHz Intel Xeon machine with 24\,GB of memory.

First, we consider running time with respect to size of graphs ($\abs{E}$).
We generate synthetic Barab\'{a}si-Albert graphs with exponentially increasing sizes.
Fraction of reports is fixed to 10\%. 
The result is shown on the left-side of Figure~\ref{fig:scalability}.
On the right side of the figure, we consider running time with respect to the fraction of reports on graph \arx. 
Both plots demonstrate \alggreedy scales roughly linearly with respect to either $\abs{E}$ or proportion of report.

\begin{figure}[t]

\setlength{\figwidth}{4.0cm}
\setlength{\tabcolsep}{0pt}
\begin{tabular}{rr}
  \includegraphics[width=\figwidth]{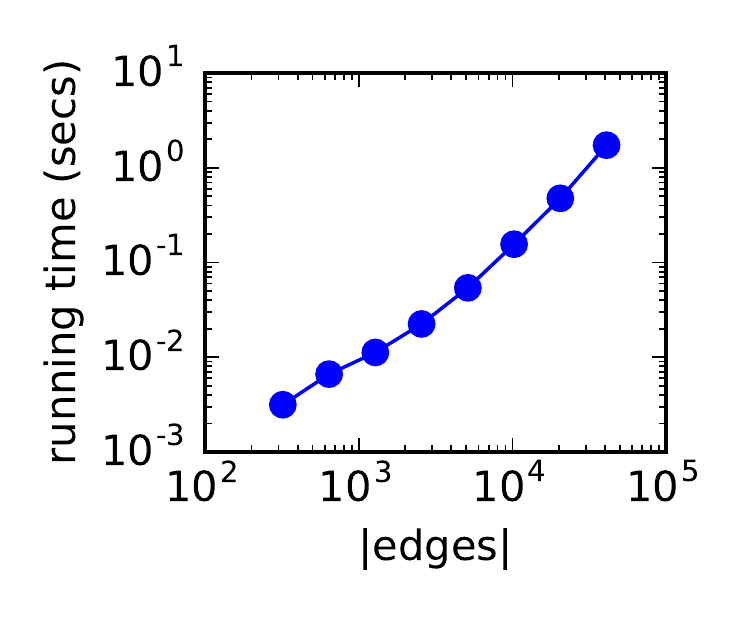} &
  \includegraphics[width=\figwidth]{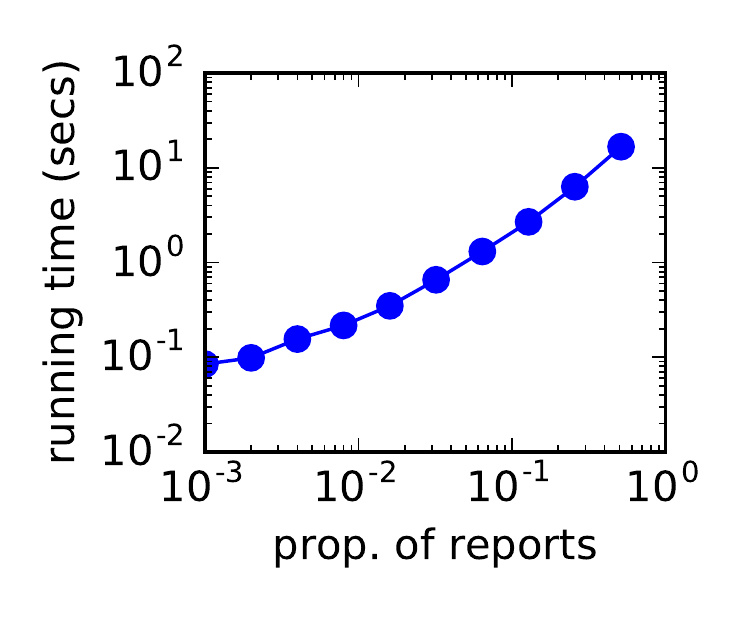} \\
\end{tabular}
\caption{Running time \alggreedy as a function of $\abs{E}$ (left figure) and fraction of reports (right figure). }
\label{fig:scalability}
\end{figure} 

\section{Conclusion}
\label{concl}

We introduce a new formulation for the cascade-reconstruction problem, 
based on a variant of the Steiner-tree problem---as 
it is common, the goal is to find a tree that spans all reported active nodes. 
The novelty of our approach is to effectively utilize temporal information of observations, 
namely, activation times.
To account for the available temporal information
we introduce temporal-consistency constraints,
requiring that all paths in the discovered tree 
should preserve the order of the observed timestamps.
For the proposed Steiner-tree problem
we present three approximation algorithms, 
which provide a trade-off between quality guarantee and scalability. 
The most efficient algorithm has linear\-ithmic running time, 
and thus, it is able to cope with very large graphs and large number of reported active nodes.

Our works opens interesting directions for future research. 
The main open problem is to close the gap
between the approximation algorithms and inapproximability lower bounds. 
Another interesting direction is to consider 
a different objective function
so as to improve the recall of the reconstructed cascade
without significant harm on precision.

\spara{Acknowledgments.}
This work has been supported by the 
Academy of Finland projects ``Nestor'' (286211),
``Agra'' (313927), and ``AIDA'' (317085), 
and the EC H2020 RIA project ``SoBigData'' (654024).

{
\balance
\bibliographystyle{IEEEtran}
\bibliography{reference}
}

\clearpage
\appendix

\section{Proofs of statements related to the problem definition}

\begin{observation}
\label{observation:root}
The optimal solution to the \Ourproblem\ problem
is a tree rooted at the reported node $(u_0, \tmin)\in \rep$. 
If there is more than one node in the set $\nodes(\tmin)$, 
any of them can be considered as a root for the optimal tree.
\end{observation}

\begin{proof}
Assume a minimal Steiner tree $T$, 
rooted in some node $\seed$.  
Let $\path$ be the path from $u_0$ to $\seed$. 
The reported nodes in $\path$ have all timestamp equal to $\tmin$. 
Let $v$ be a reported node, 
$\pathtwo$ be a path from $v$ to $\seed$, and 
$\pathtwo'$ be the path from $v$ to $u_0$. 
Then $\pathtwo'$ contains a prefix of $\pathtwo$ followed by (a part of) $\path$. 
Since the reported nodes in $\path$ have timestamp equal to $\tmin$, 
it follows that $\pathtwo'$ is an order-respecting path.
This allows us to re-root the tree from $\seed$ to $u_0$
while respecting the order.
\end{proof}

\begin{proposition}
Problem \Ourproblem is \nphard.
\end{proposition}
\begin{proof}
The standard \STproblem problem~\cite{williamson2011design}, 
with an input graph $\graph=(\nodes,\edges)$ and a set of terminal nodes $X$, can
be reduced to the \Ourproblem problem with the same graph $\graph=(\nodes,\edges)$ as the input
and reported nodes $\rep=\{(x,0)\}$, for each $x\in X$.
\end{proof}

\section{Approximation guarantee for \algclosure}

\begin{proposition}
\algclosure returns an order-respecting Steiner tree, which spans $\rep$.
\end{proposition}
\begin{proof}
First, by the for-loop defined in lines 5--9,
the resulting tree $T$ is a subgraph of $\graph$, which spans all reported nodes.
For each new reported node $u$ processed in the for-loop,
a new path $P$ is added in the subgraph \tree,
which contains only one node from $\tree$, and thus $\tree$ is indeed a tree.
Moreover, $P$ may contain only the largest time stamps in current $\tree$.
This makes $\tree$ an order-respecting Steiner tree.
\end{proof}

\begin{proposition}
\label{prop:closureapprox}
The \algclosure algorithm provides
approximation guarantee $2(1+\sqrt{3/2})\sqrt{\nrep - 1}$ for problem \Ourproblem.
\end{proposition}

To prove this result we first need to prove couple of technical lemmas.
Let us write $c = 1 + \sqrt{3/2}$ for brevity.

\begin{lemma}
\label{lem:monosplit}
Assume a sequence of $n$ numbers. 
It is possible to partition this sequence in at most $c\sqrt{n}$ monotonic subsequences
(increasing or decreasing).
\end{lemma}

Note that we are {\em not} referring to {\em consecutive} subsequences.
For example, $(5,7,9)$ is a monotonically increasing subsequence of 
the sequence $(5,8,2,7,1,9,6)$.

\begin{proof}
We will prove the lemma using induction on $n$. The result holds for $n = 1, 2, 3, 4$.
Assume $n \geq 5$, and assume that lemma holds for any sequence of length less than $n$.
Erd\H{o}s-Szekeres theorem~\cite{erdos1935combinatorial} 
states that there is a monotonic subsequence whose length is at least
$\floor{\sqrt{n - 1}} + 1$. 
By removing that monotonic subsequence we are left with a sequence having at most 
$n-\floor{\sqrt{n - 1}} - 1$ numbers.
By the induction hypothesis this reduced sequence can be 
partitioned in at most $c\sqrt{n - \floor{\sqrt{n - 1}} - 1}$ monotonic subsequences.
Thus, it is enough to prove that
\[
	c\sqrt{n - \floor{\sqrt{n - 1}} - 1} + 1 \leq c\sqrt{n}.
\]
To prove the claim note that
\[
\begin{split}
	0 &= 2 c(c - 2) - 1 \\
	&\leq \sqrt{n - 1}c(c - 2) - 1 \\
	&= c^2\sqrt{n - 1} - 2c\sqrt{n - 1} - 1 \\
	&\leq c^2\sqrt{n - 1} - 2c\sqrt{n - \sqrt{n - 1}} - 1 \\
	&= c^2n - c^2(n - \sqrt{n - 1}) - 2c\sqrt{n - \sqrt{n - 1}} - 1 \\
	&= c^2n - \pr{1 + c\sqrt{n - \sqrt{n - 1}}}^2.
\end{split}
\]
That is,
\[
	\pr{1 + c\sqrt{n - \sqrt{n - 1}}}^2 \leq c^2n,
\]
or
\[
	c\sqrt{n - \sqrt{n - 1}} + 1 \leq c\sqrt{n}.
\]
Since $\sqrt{n - 1} \leq \floor{\sqrt{n - 1}} + 1$, the lemma follows.
\end{proof}

The next two lemmas show that there \emph{exists} a directed tree in $H$ with
total weight, say $w$, less than $2c\sqrt{k - 1}$ times the number of edges in
the solution. This immediately proves Proposition~\ref{prop:closureapprox}
since that the total weight of the minimum spanning tree $A$ in \algclosure is less or equal than $w$,
and \algclosure returns a tree with number of edges bounded by the total weight of
$A$.

The first lemma establishes the bound when the reported nodes are leaves or a root, 
while the second lemma proves the general case 
when the reported nodes can be also internal nodes in the tree.

\begin{lemma}
\label{lem:puretree}
Let $\graph=(\nodes, \edges)$ be a graph and let $\rep$ be a set of reported nodes.
Consider a tree $\tree = (W, F)$, with $W\subseteq\nodes$, whose root and leaves are 
\emph{exactly} the reported nodes $\nodes(\rep)$, with the root having the smallest timestamp.
Then there is a directed tree $U = (\nodes(\rep), B)$ 
such that $(u, v) \in B$ implies $t(u) \leq t(v)$,
and
\[
\sum_{(u, v) \in B} \excldist(u, v) \leq 2 c\sqrt{\abs{\rep} - 1} \abs{F}.
\]
\end{lemma}

\begin{proof}
Write $k = \abs{\rep}$, and let $u_1, \ldots, u_{k - 1}$ be the leaves in $\nodes(\rep)$, ordered
based on a Eulerian tour. Let $\seed = u_k$ be the root of~$T$.
Consider \emph{any} subsequence of $w_1, \ldots, w_\ell$ of $u_1, \ldots, u_{k - 1}$.
Then
\[
	\excldist(\seed, w_1) + \sum_{i = 1}^{\ell - 1} \excldist(w_i, w_{i + 1}) \leq 2\abs{F},
\]
since each $\excldist(w_i, w_{i + 1})$ is upper-bounded by a path in the Euler tour
and we visit every edge in $\abs{F}$ at most twice during the tour. We can also reverse the direction
of the tour and obtain
\[
	\excldist(\seed, w_\ell) + \sum_{i = 1}^{\ell - 1} \excldist(w_{i + 1}, w_{i}) \leq 2\abs{F}.
\]
	
To create the tree $U$, start with a tree containing only the root $\seed$.
According to Lemma~\ref{lem:monosplit} we can partition $u_1, \ldots, u_{k - 1}$ to at
most $c\sqrt{k - 1}$ sequences such that the timestamps of nodes in each
subsequence is either increasing or decreasing. Let $w_1, \ldots, w_\ell$ be such
subsequence. If the time stamps are increasing, then add a path $\seed, w_1, \ldots, w_\ell$.
If the time stamps are decreasing, then add a path $\seed, w_\ell, \ldots, w_1$.
Repeat this at most $c\sqrt{k - 1}$ times.

The total distance weight of each path is $2\abs{F}$ and there are at most $c\sqrt{k - 1}$
paths, proving the result.
\end{proof}

The second lemma allows reported nodes to be non-leaves. 

\begin{lemma}
Let $\graph=(\nodes, \edges)$ be a graph and let $\rep$ be a set of reported nodes.
Consider a tree $\tree = (W, F)$, with $W\subseteq\nodes$, solving \Ourproblem.
Then there is a directed tree $U = (\nodes(\rep), B)$ such that $(u, v) \in B$ implies $t(u) \leq t(v)$,
and
\[
	\sum_{(u, v) \in B} \excldist(u, v) \leq 2 c\sqrt{\abs{\rep} - 1} \abs{F}.
\]
\end{lemma}

\begin{proof}
We will prove this by induction over $\abs{V}$.
If there are no intermediate reported nodes in $T$, we can apply Lemma~\ref{lem:puretree} directly.
Assume there is at least one, say $v$, intermediate reported node.  Let $T_1$ be the tree corresponding
to the branch rooted at $v$. Let $T_2$ be the tree obtained from $T$ by deleting $T_1$,
but keeping $v$. Let $\rep_1$ be the reported nodes in $T_2$ and let $\rep_2$ be the reported nodes in $T_2$.
Apply Lemma~\ref{lem:puretree} to $T_1$ and $T_2$, obtaining $U_1 = (\nodes(\rep_1), B_1)$ and $U_2 = (\nodes(\rep_2), B_2)$, respectively.
We can join these two trees at $v$ to obtain a joint tree $U$. This tree respects the time constraints. 
Moreover,
\[
\begin{split}
	\sum_{(u, v) \in B} \excldist(u, v) & = \sum_{(u, v) \in B_1} \excldist(u, v) + \sum_{(u, v) \in B_2} \excldist(u, v) \\
	& \leq 2 c\sqrt{\abs{\rep_1} - 1} \abs{F_1} + 2 c\sqrt{\abs{\rep_2} - 1} \abs{F_2} \\
	& \leq 2 c\sqrt{\abs{\rep} - 1} (\abs{F_1} + \abs{F_2}) \\
	& = 2 c\sqrt{\abs{\rep} - 1} \abs{F}, 
\end{split}
\]
proving the result.
\end{proof}

\section{Approximation guarantee of \alggreedy}

\begin{proposition}
Algorithm \alggreedy yields a $\nrep$-approximation guarantee for the \Ourproblem problem.
\end{proposition}
\begin{proof}
Consider an input graph $\graph=(\nodes,\edges)$, 
a set of reported nodes \rep, and seed \seed. 
Let $\tree^*$ be the optimal order-respecting tree that covers all reported nodes in $\nodes(\rep)$, 
and $\tree$ be the tree returned by \alggreedy.
Let $\cost(\tree^*)$ be the cost of optimal tree $\tree^*$, 
and $\cost(\tree)$ the cost of \tree.
Let $\path(\seed, u)$ be the shortest order-respecting path from \seed to $u$ in \graph.
For each reported node $u$ in $\nodes(\rep)$ 
we have $\cost(\path(\seed, u))\leq \cost(\tree^*)$. 
In each iteration of \alggreedy we add a path $\extexclpath(v,u)$ with 
$\extexcldist(v,u)\leq\cost(\path(\seed, u))$. 
It follows
\[
\cost(\tree) \leq \sum_{u\in\nodes(\rep)} \extexcldist(v,u) \leq 
\sum_{u\in\nodes(\rep)} \cost(\path(\seed, u)) \leq \nrep\,\cost(\tree^*).
\]
\end{proof}

\section{Approximation guarantee of \algmst}

\begin{proposition}
Algorithm \algmst yields a $\nrep$-appro\-xi\-ma\-tion guarantee for the \Ourproblem problem.
\end{proposition}

\begin{proof}
Let $\tree$ be the tree discovered by the algorithm \algmst.
We order the nodes in $\nodes(\rep)$ based on the visiting by the algorithm, 
$\nodes(\rep) = s_1, \ldots, s_k$, 
We write $\tree_i$ for the (sub)tree
that has been formed right after $s_i$ is added to the queue $Q$.
We use $p_i$ to denote the depth of the tree $\tree_i$.

Let $T^*$ be the optimal steiner tree, 
and let $T^*_i$ be a subtree of $T^*$ 
containing only branches with leaves from $s_1, \ldots, s_i$.

We claim that $p_i \leq \abs{E(T^*_i)}$. 
The proposition follows immediately from this
claim since 
\[
	\abs{E(T)} \leq \sum_{i = 1}^k p_i \leq \sum_{i = 1}^k \abs{E(T^*_i)} \leq k \abs{E(T^*)}.
\]
To prove the claim, we use induction. The result holds trivially for $i = 1$.
Assume it holds for $i - 1$.
Let $j$ be the largest index such that $t(s_j) < t(s_i)$. If no index exist,
then set $j = 1$. 

Let $q$ be the length of the path in $T^*$ from $s_i$ to the largest ancestor terminal, say
$s_\ell$, that is either a root or have a genuinely smaller time stamp.

By induction $p_j \leq \abs{E(T^*_j)}$, and since $s_\ell \in V(T^*_j)$, 
the tree $\tree_j$ can only grow by $q$ in depth before we visit $s_i$,
that is, $p_i \leq p_j + q \leq \abs{E(T^*_j)} + q$. Note that $T^*_j$
does not contain the path from $s_i$ to $s_\ell$. This implies that $\abs{E(T^*_j)} + q \leq \abs{E(T^*_i)}$, proving the claim.
\end{proof}


\end{document}